\def\fullversion{} 

\newcommand{\fullonly}[1]{\ifdefined\fullversion#1\fi}
\newcommand{\confonly}[1]{\ifdefined\fullversion\else#1\fi}

\documentclass[a4paper,USenglish,autoref]{lipics-v2021}
\nolinenumbers 
\fullonly{
  \hideLIPIcs
  \pdfoutput=1
}
\confonly{
  \relatedversiondetails[cite={BHW2026a-full}]{Full Version}{https://arxiv.org/abs/2608.12946}

  \EventEditors{Ioannis Chatzigiannakis, Andrea Vitaletti, Keren Censor-Hillel, and William K. Moses Jr.}
  \EventNoEds{4}
  \EventLongTitle{40th International Symposium on Distributed Computing (DISC 2026)}
  \EventShortTitle{DISC 2026}
  \EventAcronym{DISC}
  \EventYear{2026}
  \EventDate{November 9--13, 2026}
  \EventLocation{Rome, Italy}
  \EventLogo{}
  \SeriesVolume{397}
  \ArticleNo{12}
}

\usepackage{aliascnt}
\usepackage{hyperref}
\usepackage[noresetcount,vlined]{pwalgorithm}
\usepackage{pwmath}
\usepackage[compress]{cite} 

\usepackage[textsize=tiny,disable%
]{todonotes}

\newcommand{\dbtodo}[1]{\todo[color=cyan]{#1}}
\newcommand{\dbtodoi}[1]{\todo[inline, color=cyan]{#1}}

\newcommand{\linerefrange}[2]{lines~\ref{#1}--\ref{#2}}

\newcommand{\defop}[1]{%
    \expandafter\newcommand\csname #1\endcsname[1]{\text{\FuncSty{#1(\ensuremath{##1})}}}%
}
\defop{Read}
\defop{CAS}
\defop{FAI}
\defop{FAD}
\defop{LL}
\defop{SC}
\defop{CL}
\defop{VL}

\defop{RL}
\defop{Store}

\SetKwData{val}{val}
\SetKwData{Tag}{tag}

\title{Efficient Randomized LL/SC that Preserves History Independence}
\author{Dante Bencivenga}{University of Calgary, Canada}{drbenciv@ucalgary.ca}{https://orcid.org/0000-0002-4481-7851}{Natural Sciences and Engineering Research Council of Canada (NSERC) Canada Graduate Scholarship -- Doctoral (CGS-D) 578918-2023/198 and University of Calgary Eyes High Doctoral Recruitment Scholarship}
\author{Homa Habashi}{University of Calgary, Canada}{homa.habashi@ucalgary.ca}{https://orcid.org/0009-0002-8671-5615}{}
\author{Philipp Woelfel}{University of Calgary, Canada}{woelfel@ucalgary.ca}{https://orcid.org/0000-0002-7847-4631}{Natural Sciences and Engineering Research Council of Canada (NSERC) reference number  RGPIN-2025-04233}
\authorrunning{D. Bencivenga, H. Habashi, and P. Woelfel}
\Copyright{Dante Bencivenga, Homa Habashi, and Philipp Woelfel}

\begin{document}

\renewcommand*{\subsectionautorefname}{Section}
\renewcommand*{\subsubsectionautorefname}{Section}

\maketitle

\begin{abstract}
  We study the fundamental problem of implementing $m$ linearizable LL/SC objects with constant expected step complexity in a system of $n$ processes, using bounded base objects commonly available in hardware.
  Assuming that each process may have at most $\tau$ outstanding \LL{} operations, the best known deterministic algorithm requires $\Omega(n^2\tau + m)$ base objects (CAS objects and registers)~\cite{BW2020a}.
  Previously, no comparable randomized algorithm was known.

  By employing randomization and FADD objects in addition to CAS objects, we obtain a space bound of $O(n\tau+m)$ against the weak adaptive adversary.
  For $m=O(1)$ this matches a lower bound for algorithms using CAS objects and registers~\cite{AW2015a}.

  In addition, our object can be employed by quiescently history-independent (QHI) algorithms:
  Whenever no operation on the object is pending and no process has an outstanding \LL{} operation, its internal memory state is uniquely determined by the values of the $m$ LL/SC objects.
  \dbtodo{Changed `assumes' to `uses' $\Theta(m)$ objects, since `assumes' sounds wrong to me.}
  An important application is a recent QHI dynamic hashing algorithm, which uses $\Theta(m)$ hardware LL/SC objects to maintain a hash table of size $m$~\cite{ABFOS2025a}.
  But LL/SC is not available in hardware, and prior to our work no wait-free or efficient lock-free software implementation of LL/SC with similar properties was known.
  Our work demonstrates that one can actually implement the hashing algorithm on available hardware, without an asymptotic increase in step and space complexity, under the reasonable assumption that $m=\Omega(n)$.
\end{abstract}

\fullonly{
  \setcounter{tocdepth}{2}
  \tableofcontents
}

\section{Introduction}
Compare-And-Swap (CAS) is one of the most fundamentally and practically important shared memory primitives.
A CAS object supports the standard \Read{} operation, as well as a \CAS{old,new} operation that writes $new$ to the object if its current value equals $old$.
If this condition holds, the operation succeeds and returns \True; otherwise, it fails, leaves the object unchanged, and returns \False.

Unfortunately, CAS suffers from the ABA problem, meaning there is no guarantee that the value of the object remains unchanged between a \Read{} returning $v$ and a subsequent successful \CAS{v,v'}.
Being able to detect such ABAs is critical for the correctness of many concurrent algorithms.

Load-Linked/Store-Conditional (LL/SC), is the ABA-free counterpart of CAS\@.
An LL/SC object $L$ supports two operations, \LL{} and \SC{}.
The \LL{} operation returns the current value of the object, while \SC{v} attempts to store $v$ into the object.
The store succeeds if and only if no other successful \SC{} has been performed since the calling process's most recent \LL{}.
Otherwise, the operation fails and the object's value remains unchanged.
In any case, the \SC{} operation returns a Boolean value indicating success or failure.

Because LL/SC avoids the ABA problem, it overcomes a major limitation of CAS\@.
Consequently, many shared-memory algorithms rely on LL/SC instead of CAS~\cite{Bar1993a,Her1993,ADT1995,ST1995a,Moi1997a,CJT1998,Moi2001,Jay2002,Jay2003a,Jay2005a,BCHT2016,Bashari2021}.
Unfortunately, unlike CAS, LL/SC is not available on common hardware architectures, motivating numerous software implementations~\cite{IR1994a,AM1995a,AM1995b,Moi1997b,JP2003a,LMS2003a,DHLM2004,Mic2004,JP2005a,JP2005b,JP2005c,AW2016a,BW2020a,JJJ2023a,NW2024a}.

In this paper we present a new efficient randomized implementation of a collection of multiple LL/SC objects that improves on the state of the art in two key ways:
First, it achieves a lower space complexity compared to existing LL/SC implementations built from bounded shared memory primitives that are commonly available in hardware, while also achieving constant expected step complexity.
Second, the implementation can be used in algorithms that assume hardware LL/SC to achieve quiescent history independence~\cite{ABFOS2024a}.
No other known efficient wait-free algorithm has this property.
Our algorithm is \emph{Las Vegas} in the sense that randomization affects only running time and not correctness or history independence.

We will now describe each of these results, which are both derived from the same algorithm.

\subsection{Improved Space Complexity}
A simple way to implement LL/SC using CAS is to augment each CAS object with an unbounded sequence number that increments with every successful \SC{} operation.
Significant efforts have been made to avoid using unbounded sequence numbers.
We will assume throughout that shared memory hardware objects are of bounded size (in particular, unbounded sequence numbers cannot be used), and will refer to the number of such objects used as the space complexity.
The best algorithms that implement a single LL/SC object from commonly available hardware primitives achieve constant step complexity and $O(n)$ space complexity, where $n$ is the number of processes~\cite{BW2020a,JP2003a}.
This matches a time-space tradeoff result, stating that any implementation from $S$ registers and CAS objects with step complexity $T$ satisfies $S \cdot T = \Omega(n)$~\cite{AW2015a}.

Many algorithms require multiple LL/SC objects.
The naive approach of composing $m$ independent LL/SC implementations leads to an undesirably high space complexity of $\Theta(n\cdot m)$.
Better space complexities are often in terms of the maximum number of \emph{outstanding \LL{} operations}, which informally are \LL{} operations that have been executed but have not yet been paired with subsequent \SC{} operations (see \autoref{sec:outstanding_ll} for a formal discussion).
The LL/SC algorithm by Blelloch and Wei~\cite{BW2020a} improves on the naive approach by implementing $m$ LL/SC objects using $O(n^2\tau+m)$ bounded CAS objects and registers, where $\tau$ is an upper bound on the number of outstanding \LL{} operations per process.
As the authors argue, $\tau$ is typically a small constant.
However, if $\tau$ were to asymptotically exceed $m/n$, then the space complexity of their algorithm would become worse than that of the naive composition.
We are not aware of any comparably efficient randomized solutions.

Our randomized algorithm also implements $m$ LL/SC objects and has constant expected step complexity (against the weak adaptive adversary), while improving upon the best known space bounds.
If the total number of outstanding \LL{} operations across all LL/SC objects is known in advance to be at most $k$, then $O(k+m)$ bounded hardware primitives suffice.
We use CAS objects and fetch-and-add (FADD) objects, both of which are widely available in hardware.
For $\tau$ outstanding \LL{} operations per process, this yields a space improvement from the $O(n^2\tau+m)$ base objects of Blelloch and Wei~\cite{BW2020a} to $O(n\tau+m)$.
Since $\tau\leq m$, our algorithm never needs asymptotically more hardware objects than the naive composition of $m$ optimal LL/SC implementations.

On the downside, our algorithm requires storing bounded tags of size $\ceil{\log (\tau n)} + 1$ within CAS objects, in addition to the LL/SC value.
In contrast, the algorithm of Blelloch and Wei can implement multi-width LL/SC using only pointer-width CAS objects.
In practice, when the stored value is a $64$-bit heap pointer, the $16$ high-order bits and $3$ low-order bits are typically zero on modern systems, leaving $19$ bits to store a tag~\cite{Cha2022a}.
Blelloch and Wei~\cite{BW2020a} argue that $\tau \leq 2$ in practice, although a later history-independent hash table~\cite{ABFOS2025a} uses $\tau = 3$; we do not know of any algorithm using LL/SC (QHI or not) which requires $\tau$ larger than $3$.
This implies that tagged pointers can support up to $2^{18}/3$ processes for our implementation, which is far above currently practical thread counts.
In cases where all $64$ bits are required (e.g., double-precision floats), most architectures support double-width CAS for $128$ aligned bits.

Additionally, our implementation supports a \CL{} (``clear link'') operation, enabling a process to ``undo'' a preceding \LL{} (clearing the calling process's outstanding \LL{} operation on a given LL/SC object), which helps to reduce the space complexity.
\dbtodo{Removed the reference to `open link' here due to reviewer C's comment.}
It also supports a \VL{} (``validate link'') operation, allowing a process to verify whether an \SC{} would succeed at that point, without changing the object's state or causing future \SC{} calls to fail.
This is often helpful in algorithms involving multiple LL/SC objects.

\todo{Add something about the techniques being used. They are in some way standard, but there are some new aspects to it. E.g., the combination of bounded FAI and CAS. Can this be explained in one or two sentences?}
\dbtodo{Added three sentences below.}
Our algorithm's basic structure is similar to early bounded-tag LL/SC implementations such as that of Anderson and Moir~\cite{AM1995a} and Moir~\cite{Moi1997b}.
The space complexity bottleneck of many such algorithms is due to needing to track which tags can be safely reused, and choosing one in constant time.
This is difficult to accomplish deterministically, so we instead employ randomization to choose new tags for \SC{}, and use FADD objects to count how many outstanding \LL{} operations have read each tag.
These usage counts prevent ABAs of the \Tag fields by ensuring that the new tags chosen by \SC{} calls have not been read by any outstanding \LL{} operation.
The counts are inspired by reference counting techniques~\cite{Val1995a,Gre1999a,DMMS2002a}, but since we do not use them to free memory, we can circumvent the usual complications involved in safe memory reclamation.

\subsection{History Independence}\label{sec:intro:hi}
In a history-independent data structure, an attacker obtaining a snapshot of the memory cannot learn anything about the past operations applied to the data structure, beyond what is implied by the current abstract data structure state.
This is an important privacy property, which has been extensively studied in sequential computing
~\cite{Mic1997a,NT2001a,HHMPR2002a,BP2006a,ABHVW2004a,HHMPR2005a,BG2007a,NSW2008a,Gol2009a,Gol2010a,BBJKMPSZ2016a}.
While there is a history-independent hash table in a partially concurrent setting~\cite{SB2014a}, only recently has a systematic study of history independence in the concurrent setting been initiated by Attiya, Bender, Farach-Colton, Oshman, and Schiller~\cite{ABFOS2024a}.

Common techniques, such as using sequence numbers or tags to avoid ABAs, timestamps to order operations, or process identifiers to support synchronization, cannot easily be employed in history-independent data structures, as they encode traces of past executions.
LL/SC exemplifies this: To deal with ABAs, an \SC{} operation succeeds if and only if there was no other successful \SC{} since the calling process's preceding \LL{}, and therefore depends on a record of the past.
If LL/SC were implemented in hardware, such metadata could not be retrieved by an attacker through reading shared memory.
However, software implementations must store it in shared memory, potentially exposing it.

We consider \emph{quiescent history independence} (QHI), where the attacker can access memory only when no operations are in progress~\cite{ABFOS2024a}.
The authors of~\cite{ABFOS2024a} also consider \emph{perfect} history independence (PHI), in which the attacker can access memory at any time, but that often makes the attacker too strong:
As they prove, PHI requires impractically large base objects to implement many types of concurrent objects.
A slightly stronger property than QHI is \emph{state} QHI~\cite{ABFOS2024a}, where the attacker can access memory during operations that cannot change the state of the abstract data structure.

A recently proposed randomized hash table~\cite{ABFOS2025a} satisfies state QHI while maintaining strong efficiency: a table of size $m$, storing elements from a universe of size $u$, uses $O(m)$ shared memory objects of size $2\log u+O(1)$ bits each.
It is lock-free with expected amortized step complexity $O(c)$, where $c$ bounds the number of concurrent operations on the same element.

This algorithm critically depends on LL/SC\@.
Since this synchronization primitive is not available in hardware, the hashing algorithm cannot be efficiently implemented on real systems without software support of some form of LL/SC that preserves the history independence of the algorithm.

Observe that even if an LL/SC implementation is itself history-independent, it cannot be trivially used to replace atomic LL/SC\@:
The state of an LL/SC object generally comprises \emph{context} (see \autoref{sec:seq_spec}), describing the set of processes whose next \SC{} should succeed---we say that such processes have \emph{open links} to the object.
Thus, if an LL/SC object has open links in a quiescent state, then an attacker can obtain information about the context.
But algorithms employing atomic LL/SC, such as the hashing algorithm in~\cite{ABFOS2025a}, may assume that that context is not exposed.\footnote{Note that this assumption is only justified if an attacker has access to merely the shared memory, and cannot perform, for example, \SC{} operations. See also the discussion in Section~\ref{sec:prelims:hi}.}
Thus, when using implemented LL/SC, it may be necessary to ensure that no process has an open link in a quiescent state.
To facilitate that, Attiya et al.~\cite{ABFOS2024a} suggested augmenting the LL/SC specification with an operation that allows processes to clear their open links before a quiescent state is entered. 
The authors called this operation \emph{release} or \RL{}, but for consistency with earlier literature (e.g.,~\cite{Moi1997b,BW2020a}) we call it \emph{clear link} or short \CL{}.

The authors showed that there exists a history-independent implementation of LL/SC with a \CL{} operation from a single CAS object.
(The core algorithm is essentially the same as that of Israeli and Rappoport~\cite{IR1994a}.)
In fact, the algorithm satisfies \emph{perfect} history independence. 
This is an important existential result, used by the authors to show that a wait-free history-independent universal construction exists.
However, the algorithm is not efficient and only lock-free (even though, perhaps surprisingly, the resulting universal construction is wait-free):
Each successful \SC{} operation can prevent $n-1$ \LL{} operations from making progress, leading to a worst-case amortized step complexity of $\Omega(n)$. 
Moreover, the underlying CAS object must be able to store $n$ bits of metadata in addition to the value of the LL/SC object.
Hence, the algorithm does not scale well with the number of processes.

In this work, we define \emph{QHI-preserving} LL/SC\@.
That property is weaker than QHI, but sufficiently strong to preserve QHI of an algorithm in essentially the same way as QHI LL/SC does.
We then show that our efficient implementation of multiple LL/SC objects has this property.
It is the first efficient QHI-preserving LL/SC implementation from primitives that are commonly available in hardware.

Our LL/SC implementation can be used in the universal construction of Attiya et al.~\cite{ABFOS2024a} while preserving state QHI\@.
Since our methods are randomized wait-free rather than just lock-free, their construction also becomes simpler (see \autoref{sec:llsc_in_ABFOS}). 
The required bits of metadata per represented LL/SC object is reduced from $n$ to $\log n + O(1)$.
However, our implementation requires additional base objects, in this case $5 n + 1$ FADD objects, each with $\log n + O(1)$ bits.

The hash table of~\cite{ABFOS2025a} retains its asymptotic step and space complexity when substituting hardware LL/SC with our implementation, under the reasonable assumption that $m=\Omega(n)$.
\dbtodo{Rephrased as per reviewer C.}
But, using our LL/SC implementation, the hash table is not state QHI anymore, and the size of its required CAS objects increases by an additive term of $\log n$ to $2\log u+\log n+O(1)$.
To preserve the space bounds, one needs to insert \CL{} operations in appropriate places, so that outstanding \LL{} operations are cleared whenever they are not needed anymore, as we discuss in \autoref{sec:prelims:hi}.
(This technique was introduced in the history-independent universal construction of~\cite{ABFOS2024a}.)



\subsection{Other Related Work}
Almost all prior implementations of $m$ LL/SC objects either do not achieve constant worst-case step complexity~\cite{DHLM2004,IR1994a}, use unbounded sequence numbers or tags~\cite{JP2005a,JP2005c,Mic2004,Moi1997b,JJJ2023a}, or require at least $\Omega(n\cdot m)$ base objects.
The sole exception that avoids all these drawbacks is the algorithm by Blelloch and Wei~\cite{BW2020a}, discussed earlier.
Other work addresses properties such as strong linearizability~\cite{NW2024a}, durability~\cite{JJJ2023a}, or independence of the number of processes in the system~\cite{JP2005c,DHLM2004,JJJ2023a}.
A time-space tradeoff result states that a single 1-bit LL/SC object implemented from $s$ bounded CAS objects and registers has step complexity $\Omega(n/s)$~\cite{AW2015a}.
\dbtodo{Changed $s=O(t)$ to $s=O(1)$.}
Algorithms achieving optimality with respect to this tradeoff are known for $s=\Theta(n)$~\cite{JP2003a} and $s=O(1)$~\cite{AW2015a}.
Our algorithm is also optimal for $s=\Theta(n)$ (and $m=1$).


\section{Preliminaries}\label{sec:prelims}
We use the standard asynchronous shared memory model, where $n$ processes communicate by performing atomic shared memory operations on CAS objects and FADD objects.
In fact, instead of the more general FADD, fetch-and-increment/decrement (FAID) suffices.
This object stores an integer in the range $\set{0, \dots, M-1}$ for some positive integer $M$, and supports the operations \FAI{}, \FAD{}, and \Read{}.
Each operation returns the object's current value, and \FAI{} and \FAD{} respectively increment and decrement it. 
(Our algorithm guarantees that no overflow occurs.)

Our algorithm is randomized, meaning that at any point a process can generate a (private) random number by performing a \emph{coin flip step}.
The adversary determines the order in which processes take steps, by deciding at any point during the execution, the next process to take a step.
We assume the \emph{weak adaptive adversary}, where a process always performs its next shared memory operation immediately after a coin flip step.
I.e., while the weak adaptive adversary knows all past random decisions and adapts to them, it cannot intervene between a process's coin flip and its following shared memory step.
Even though we employ randomization, all correctness properties such as linearizability and QHI are achieved deterministically, i.e., they hold for all executions, regardless of the random decisions made.

\subsection{Sequential Specification}\label{sec:seq_spec}
An \emph{abstract object} (also known as \emph{type}) is a quintuple $(Q,q_0,I,R,\delta)$, where $Q$ is a set of states, $q_0$ is the initial state, $I$ and $R$ are sets of operations and responses, respectively, and $\delta:Q\times I\to Q\times R$ is the transition function.
This quintuple is also called the \emph{sequential specification} of the object, as it defines the set of possible sequential executions: Whenever an object is in state $q$ and operation $i$ is invoked, $\delta(q,i)=(q',r)$ defines the operation's response $r$ and resulting new state $q'$.

To specify the abstract object LL/SC, we follow the definition in~\cite{ABFOS2024a}, which itself is based on~\cite{JJJ2023a}.
(In~\cite{ABFOS2024a}, the specified object is called \emph{context-aware} LL/SC\@.
But context awareness is inherent to the specification, so we simply call it LL/SC.)
Each state is a pair $(v,\mathcal{L})$, where $v$ is the object's \emph{value} and $\mathcal{L}$ is an $n$-bit sequence $(\ell_1,\dots,\ell_n)$, called \emph{context}.
\dbtodo{We are already using $L$ to mean the array of LL/SC objects in the implementation, but I think the two $L$'s are never in scope at the same time as each other.
How about $O$ for open links?
And should we rename `context' to just `open links', with the new definition at the bottom of this paragraph?}
Initially, $\ell_1=\cdots=\ell_n=0$, and the initial value of $v$ can be defined arbitrarily.
\todo{It would be nice if we could say that $p$ has an outstanding link whenever $\ell_p=1$, but given the definition later, this would not be linearizable.
We think that if we did not want to achieve the global space bound, then we could put outstanding links in the sequential spec, but not sure if this can be changed now.
Probably the global space bound is less important, though.
}
\dbtodo{I think it's too much work to change back to a local bound.
However, with my clarification of definitions we can say that $p$ does have an open link whenever $\ell_p=1$.}
We say that $p$ has an \emph{open link} to the LL/SC object whenever $\ell_p = 1$.

\dbtodoi{Should we rephrase the following paragraph to use the phrase `open link' as defined here, or keep it as is?}
\todo{Keep it.}
A \Read{} returns the value $v$ of the object.
An \LL{} (``load-linked'') call by process $p$ sets $\ell_p$ to 1 and returns the value $v$.
An \SC{x} (``store-conditional'') operation by process $p$ succeeds if $\ell_p=1$, in which case it changes the state of the object to $(x,(0,\dots,0))$ and returns $\True$.
If $\ell_p=0$, then the operation fails and returns \False without changing the state of the object.
Operation \CL{} (``clear link'') by process $p$ simply resets $\ell_p$ to 0, and does not return anything.
Finally, operation \VL{} (``verify link'') by process $p$ returns a Boolean value that is true if and only if $\ell_p=1$.

In this paper, we implement a collection of $m$ LL/SC objects.
Thus, the abstract object's state is the cross product of the states of $m$ individual LL/SC objects.
Each operation now has one additional parameter that indicates which of the LL/SC objects in the collection is affected.
For example, operations \LL{i} and \SC{i,x} behave as \LL{} and \SC{x} operations defined above, affecting only the $i$-th LL/SC object in the collection.


\subsection{Outstanding LL() Operations}\label{sec:outstanding_ll}
\todo{Can we change that, too, and add an explanation about the difference between a ``really'' outstanding LL and what is being defined here? Problem is that it's used in an ``atomic'' sense of having pending links in the intro. Or maybe change the intro to ``pending'' link and add that to the sequential spec, then explain here the difference.}
\dbtodo{I've gone ahead with the second option.
I like the idea of `open link' (since you use that term more) to refer to the sequential specification, as the abstract concept of an \LL{} whose next \SC{} or \VL{} is poised to succeed if executed.
I believe we came up with the term `outstanding \LL{} operation' to specifically refer to a process which has performed an \LL{}, and, as far as it locally knows, its next \SC{} or \VL{} might succeed.
(There is however a wrinkle in an \LL{} operation which returns in \autoref{line:ll:final_return} knows that its next \SC{} or \VL{} will fail.
However, that is specific to our implementation, and the term `outstanding \LL{}' should not refer to our specific implementation.)}
A process begins to have an outstanding \LL{} operation on object $i$ (or, equivalently, an outstanding \LL{i}) when it invokes an \LL{i} operation, and it ceases to have an outstanding \LL{i} when it completes one of the following operations: a \SC{i, \cdot}, \CL{i}, or a \VL{i} operation that returns \False.
In particular, a process either has an outstanding $\LL{i}$ or it does not; it cannot, e.g., have more than one outstanding \LL{} on the same object $i$.

Outstanding \LL{} operations are related to but distinct from open links.
Open links are defined by the sequential specification and refer to the abstract state of the LL/SC object.
Outstanding \LL{} operations are defined by the local knowledge of each process, and change at invocation and response points of operations.
If a process has an open link to object $i$ while it is idle, then it also has an outstanding \LL{i}.
However, a process may have an outstanding \LL{$i$} without an open link to object $i$.
In particular, a successful \SC{$i$} operation by process $p$ clears open links to object $i$ from \emph{all} processes, but \emph{only} clears $p$'s outstanding \LL{i}---other processes retain their outstanding \LL{i} until they `discover' that their open link has been cleared.

The implementation uses a known bound $k$ on the number of total outstanding \LL{} operations at any given point in the execution.
Under a bound of $\tau$ outstanding \LL{} operations per process, we would thus have $k = \tau n$.

\subsection{History Independence}\label{sec:prelims:hi}
We consider a linearizable implementation of an abstract object together with a function $h$ (called linearization function), which maps every execution on that object to a linearization.
A linearization~\cite{HW1990a} $h(E)$ of a concurrent execution $E$ is a sequential execution of all operations that have responded in $E$ and some operations that are pending (i.e., invoked but not yet responded) at the end of $E$, such that
\begin{itemize}
  \item For any two calls $a$ and $b$ such that call $a$ responds before call $b$ is invoked in $E$, $a$ appears before $b$ in the linearization $h(E)$,
  \item The invocation arguments and response values (if applicable) of each operation in $E$ matches those in $h(E)$, and
  \item The sequential execution $h(E)$ is consistent with the object's sequential specification.
\end{itemize}
An implementation for which a linearization function exists is called linearizable.
\todo{Now we should probably define linearizability and linearization.}
\dbtodo{Added a definition, but should I also define linearization points in there?}

We assume that at certain points in time during an execution, an attacker gets access to the shared memory.
The types of history independence defined in~\cite{ABFOS2024a} differ mostly in their restrictions on when those attacks can occur.
Formally, we can assume a set $S$ of finite executions, so that attacks can only occur at the end of each execution $E\in S$.
The implementation is history-independent, if for every execution $E\in S$, the state of the shared memory resulting from $E$ is uniquely determined by the state of the abstract object resulting from the linearization $h(E)$.
(I.e., if two executions $E,E'\in S$ yield different shared memory states, then the sequential executions $h(E)$ and $h(E')$ must also yield different abstract object states.)

We focus on \emph{quiescent history independence} (QHI), where an attack can occur only when the object is in a quiescent state (i.e., no operation is pending).
Hence, $S$ is the set of all executions that end in such a quiescent state.

\dbtodoi{I've rephrased the following two paragraphs a bit.
I'm trying to emphasize the abstract states involved, and what they imply regarding HI/QHI.}
Now consider a simple object, say a dictionary with operations \FuncSty{insert()}, \FuncSty{delete()}, and \FuncSty{lookup()}.
Suppose that in the implemented \FuncSty{lookup()} operation, some process $p$ calls \LL{} to read the value of an LL/SC object $L$, and then $p$ never follows up with an \SC{} call.
Consider an execution $E$ that consists of only one such complete \FuncSty{lookup()} operation, and let $h(E)$ be its linearization.
At the end of $E$, the context $(\ell_1,\dots,\ell_n)$ of $L$ has $\ell_p=1$, because $p$ has an open link to $L$.
But if the process hadn't called \FuncSty{lookup()}, $\ell_p$ would be 0.
Thus, the abstract states of $L$ at the end of $E$ and at the end of the empty execution are different.
On the other hand, the linearizations of the two executions are the empty execution and one where only one \FuncSty{lookup()} is being performed.
Thus, both yield the same abstract object state \emph{of the dictionary} (that is, the empty dictionary) even though the abstract state \emph{of the LL/SC object} differs (in the value of $\ell_p$).
Since the two executions end in quiescent configurations with different memory states (namely, the abstract state of $L$), QHI is violated.

The QHI hashing algorithm in~\cite{ABFOS2025a} does not perform such \CL{} operations, and in fact suffers exactly from the above problem:
After an execution in which only a single \FuncSty{lookup()} operation is performed, the calling process has an open link to an LL/SC object, whereas in the empty execution it does not.

\dbtodo{Technically this next sentence is true for the other variants of history independence.}
However, if one defines the state of an LL/SC object as its value \emph{without} context, then a QHI LL/SC implementation would preserve QHI for algorithms using it, including the hashing algorithm in~\cite{ABFOS2025a}.
We will denote such an LL/SC object as \emph{context-oblivious}.
It depends on the attack model whether context-oblivious LL/SC can leak information in such an algorithm:
If an attacker can only observe (read) shared memory, then it can indeed not learn more than the state represented by an abstract context-oblivious LL/SC\@.
If, on the other hand, the attacker can run code on shared memory, then it can use \VL{} or \SC{} operations to expose open links to LL/SC objects, even if they are context-oblivious.
In any case, technically it is problematic to assume context-oblivious LL/SC for history independence, as it is impossible to provide a sequential specification where the object's state is well-defined and the outcome of an operation only depends on this state.
Moreover (and as a result of that) it is impossible to implement context-oblivious LL/SC from other primitives.
Hence, an LL/SC implementation satisfying (even perfect) history independence, such as~\cite{ABFOS2024a}, can in general not simply be used to replace context-oblivious LL/SC and still preserve history independence.

As we have argued, the abstract state of implemented LL/SC objects must include its context.
However, the abstract state of objects implemented using LL/SC typically does \emph{not} depend on LL/SC context, since that is usually an implementation detail.
It therefore follows that a QHI implementation typically has to ensure that there are no open links in quiescent configurations.
For this reason, in their universal construction, Attiya et~al.~\cite{ABFOS2024a} employ \CL{} operations, by letting each process clear its open links before responding from a method call. 

Fortunately, it is easy to modify the hashing algorithm in~\cite{ABFOS2025a} to satisfy this requirement: During each implemented operation, a process keeps track of the \LL{} operations it has called, and, at the end of the method call, clears all links by performing \CL{} operations.
Clearly this does not increase the asymptotic step complexity.
In general, this transformation can be applied to all algorithms where a process calling \LL{} in a method does not rely on that \LL{} operation in a future method call.
This is the case in the hashing algorithm, and in perhaps all but some contrived algorithms.\footnote{One such contrived example is the implementation of a QHI LL/SC object from a single QHI LL/SC object.}

In an LL/SC implementation in which the space complexity depends on the number of outstanding \LL{} operations, such as ours, this general transformation may however increase the space complexity, especially if the algorithm contains loops.
To preserve space complexity, \CL{} operations would need to be added once an outstanding \LL{} operation is no longer needed, to ensure that each process has a minimum number of outstanding \LL{} operations at any given time.
For the hashing algorithm specifically~\cite{ABFOS2025a}, it is possible to show that there are always at most three outstanding \LL{} operations per process which may be later used in \SC{} or \VL{} operations; we discuss this in \autoref{sec:llsc_in_ABFOS}.

\subsection{QHI-Preserving LL/SC}\label{sec:qhi-preserving}
As discussed above, a QHI implementation of an abstract object that uses implemented LL/SC must typically prevent quiescent configurations from having open links.
Thus, when implementing LL/SC for the purpose of using it in QHI algorithms, it makes sense to only require that no information is leaked when there are no outstanding \LL{} operations, and therefore no open links.
\dbtodo{I don't understand this `especially the case' comment, since we have already argued for why we should not allow quiescent configurations from having outstanding \LL{} operations.}
This is especially the case when this relaxed requirement can lead to improved efficiency.
In other words, for the purpose of using implemented LL/SC in QHI algorithms, there is little advantage in requiring an LL/SC object to not leak more information than just context in configurations where context is leaked anyway.

This leads us to the definition of \emph{QHI-preserving} LL/SC, where we allow arbitrary information leakage whenever there are outstanding \LL{} operations.

\todo{I'm afraid we'll have to rename your definition of quiescent state. What about ``unencumbered state''?}
\dbtodo{I've renamed it to `settled' state.}

\dbtodoi{In our implementation, the \Read{} method contains only one step so processes cannot have pending method calls on it.
I'm leaving it as is, though, since I think the intention is a more general treatment of implemented LL/SC objects.}
A collection of LL/SC objects is in a \emph{settled state} if no process has a pending method call on any object, 
\emph{and} no process has an outstanding \LL{} on any object.
Since LL/SC is a building block for other implementations, processes would usually only have an outstanding \LL{} operation when they have a pending method call for an implementation that uses LL/SC\@.
As such, most implementations can ensure that \CL{} is called to clear any outstanding \LL{} operations before an implemented method responds.
When such an implementation is in a quiescent state in the sense of no process having a pending method call, our LL/SC implementation is also in a settled state in the sense of having no outstanding \LL{} operations.

\dbtodo{Changed as per reviewer A}
An LL/SC implementation is \emph{QHI-preserving} if, for every settled state reachable in any execution, its internal shared memory state is uniquely determined by the interpreted values of the LL/SC objects.

Clearly, any QHI LL/SC implementation is also QHI-preserving, but the opposite is not necessarily true.
However, in any QHI algorithm, where in every quiescent state there are no outstanding \LL{}  operations, one can replace QHI LL/SC implementations with QHI-preserving ones, without sacrificing the QHI property of the algorithm.

If a state QHI algorithm uses our LL/SC implementation, and also avoids \LL{}, \SC{}, \CL{}, and \VL{} operations in its non-state-changing operations (i.e., it only uses the \Read{} method on LL/SC objects), then it remains state QHI\@.
This is the case for the universal construction of~\cite{ABFOS2024a}, and so using our implementation of LL/SC in that construction does preserve state QHI (see \autoref{sec:intro:hi}).
Unfortunately, for the hashing algorithm discussed earlier~\cite{ABFOS2025a}, its non-state-changing \FuncSty{lookup()} method relies on \LL{} and \VL{} operations, and so it is not state QHI under our LL/SC implementation.
As argued in \autoref{sec:prelims:hi}, since \LL{} changes the abstract state of the LL/SC object, the hashing algorithm would not be state QHI under \emph{any} LL/SC implemented from other primitives.

Note also that even though our algorithm is randomized, it is deterministically QHI-preserving: Whenever the collection of LL/SC objects is in a settled state, the memory state is in a state that is deterministically and uniquely determined by the object's abstract state.

\section{Implementation}

\subsection{High-Level Description}

We first explain how a simplified version of the algorithm works, without preserving quiescent history independence.

\dbtodo{Possible confusion: I use the word `tag' to refer to both the second field of elements of $L$, e.g., $L[i].tag$, and also to refer to the possible values that these fields can take on.
I think that it's reasonably clear from context, but perhaps I should use slightly different terminology to differentiate these two concepts?}
\dbtodo{Somewhat last minute, but I'm trying to consistently say `\Tag field' when referring to the field that stores tags, with the values themselves still being called `tags'.}
Each LL/SC object is encoded as a CAS object storing a value-tag pair.
The \Tag field stores a positive integer tag with a small, finite range (tag zero is reserved for preserving QHI), which allows \SC{} operations to detect changes since the calling process's previous \LL{} operation, and hence prevent ABAs.
Each possible tag has an associated \emph{usage count}, stored in an array of FAID objects, which count how many outstanding \LL{} operations have read that tag.
Each \SC{} operation attempts to change the object's \Tag field to a randomly chosen new tag with zero usage count, via a \CAS{} operation.
Our implementation guarantees that this \CAS{} attempt succeeds exactly if there was no other successful \CAS{} on the value-tag pair since the corresponding \LL{} operation last read the LL/SC object.

More specifically, an \LL{} operation reads the value-tag pair, and attempts to \emph{protect} the tag by incrementing its associated usage count.
The calling process re-reads the value-tag pair, and if it has changed, then it can linearize before the change, such that its matching \SC{} and/or \VL{} calls can correctly return \False.
In this case, the calling process responds without saving the value-tag pair that it read, as though it did not perform any \LL{} operation at all.
Otherwise, if the value-tag pair stays consistent after the process increments the tag's usage count, then the process locally saves the value-tag pair for use in future operations on that LL/SC object.
(In this case, the operation linearizes at the second read, since it is possible that the value-tag pair had an ABA after the first read and before the process protected the tag.)
The usage count then remains nonzero while the \LL{} operation remains outstanding, which prevents a tag from being reused for that same object.
Thus, the process can later tell whether it still has an open link for its outstanding \LL{} operation by simply checking whether the \Tag field has changed or not from its locally saved value.
Indeed, this is exactly what a \VL{} operation does.

An \SC{} operation by process $p$ first tries to find a new tag that is not protected by any outstanding \LL{} operation, including its own outstanding \LL{}.
It repeatedly reads the usage counts of randomly selected tags until it finds one with a count of zero, which is the only randomized step of the implementation.
Since there are at most $k$ outstanding \LL{} operations in total, and $2k$ possible tags, each usage count has a probability of at least $1/2$ to be zero.
(Recall that we assume the weak adaptive adversary, and therefore that $p$ reads a usage count immediately after randomly choosing a tag, without any other process taking steps in between.)
Then, $p$ attempts a \CAS{} from the previous value-tag pair to the new value and newly chosen tag, and returns whether this \CAS{} succeeds.
Note that the adversary \emph{can} schedule other process steps in between choosing a new tag and this \CAS{} attempt, potentially making other \SC{} operations choose the same tag.
However, if any other \SC{} successfully changes the value-tag pair in between these steps, then $p$'s \CAS{} attempt fails, because $p$ already protected the tag read in its previous \LL{} call, preventing it from being reused during the \SC{}\@.
(We discuss this in more detail in \autoref{sec:low-level:sc} and prove the claim in \autoref{lem:invariant}(\ref{item:invariant:protection}).)

The remaining operations are straightforward:
A \CL{} operation decrements the read tag's usage count, and clears the process's locally saved value-tag pair.
It is also called internally whenever an outstanding \LL{} operation is cleared.
A \Read{} operation returns the value part of the value-tag pair.

\subsubsection{History Independence}
To preserve quiescent history independence, we first introduce a default tag of zero, such that all LL/SC objects' \Tag fields are initialized to zero.
This default tag does not have a usage count, but we instead add an activity counter to each LL/SC object, which is a FAID object that counts the number of outstanding \LL{} operations involving that object.
This activity counter is then used to detect a settled state relative to this object, i.e., when no process has a pending method call or an outstanding \LL{} operation on that object.
The \LL{} method increments this counter, and the \CL{} method (including internal calls of it) decrements it.
The canonical state of shared memory during a settled state has all \Tag fields and all counters set to zero.
While the counters naturally reach zero in a settled state, we need extra steps to ensure that \Tag fields are also reset to zero.

The \CL{} method now takes on the responsibility of resetting an LL/SC object's \Tag field to zero, if it decrements its activity counter to zero.
Since other operations may still start working on the same object concurrently, the calling process takes similar steps to an \LL{} followed by an \SC{} operation, but only to change the \Tag field back to zero.
In particular, the random selection of a tag with zero usage count is replaced with double-checking that the activity counter is still zero, and the \CAS{} only attempts to change the \Tag field to zero from its previously read value.
If the \CL{} call detects any activity on the object, then it stops trying to reset the \Tag field, since it is no longer the `last' operation working on the object.
(I.e., another process would have incremented the activity counter again, relieving the calling process from its responsibility to reset the \Tag field).

\dbtodoi{Is this last paragraph too technical?}
Note that this added step of \CL{} resetting the \Tag field to zero does not change the abstract state of the object, and in particular it must not clear any other process's open link.
Moreover, a \CL{} operation could be poised to reset a \Tag field to zero while other operations start working on the object, even though it checks the activity counter immediately before its attempted reset.
As such, the \LL{}, \SC{} and \VL{} methods have extra logic to avoid misinterpreting a reset \Tag field as a successful \SC{} operation which cleared open links.
In particular, the \LL{} method always accepts tag zero (also because tag zero does not have a usage count to increment), reverting its increment of a usage count in case it first reads a nonzero tag and then tag zero.
If the \SC{} method's first \CAS{} from its known value-tag pair fails, it now also attempts a second \CAS{} from the old value and tag zero to the new value-tag pair, in case the \Tag field was reset during its operation.
The \VL{} method returns \True if the \Tag field has not changed \emph{or} if it has changed to zero.

\subsection{Pseudocode}

\newcommand{\HI}[1]{{\color{red}#1}}
\newcommand{\nonHI}[1]{\textcolor{black}{#1}}

We present the implementation in pages~\pageref{fig:main_methods} and~\pageref{fig:additional_methods}.
The \HI{red} lines are only needed to preserve quiescent history independence, so a simpler algorithm following only the black lines implements LL/SC with the same asymptotic time and space complexity but without preserving history independence.
(We only analyze the full algorithm in \autoref{sec:analysis}.)
For the process-local static dictionary $d$, we use the default value $\bot$ to mean that a given key does not have a value in $d$.
When referring to pairs containing a value and a tag (which elements of $L$ and $d$ store), we use the dot operators $.\val$ and $.\Tag$ to refer to the first and second fields of the pair, respectively.

\begin{figure}\label{fig:main_methods}
    \begin{invisiblebox}{\KwSty{Constants:}}
    \begin{itemize}\fontsize{9.5}{11.5}\selectfont
      \item $U$ is the set of values that LL/SC objects can have
      \item $m$ is the number of LL/SC objects
      \item $n$ is the number of processes
      \item $k$ is the maximum number of outstanding \LL{} operations that processes can have in total
    \end{itemize}
  \end{invisiblebox}

  \begin{invisiblebox}{\KwSty{Shared Data:}}
    \begin{itemize}\fontsize{9.5}{11.5}\selectfont
      \item $L[i]$ for $i \in \set{1, \dots, m}$ is a CAS object storing a pair in $U \times \set{\HI{0,\,}1, \dots, 2k}$
      \HI{\item $C[i]$ for $i \in \set{1, \dots, m}$ is a FAID object storing a number in $\set{0, \dots, n}$}
      \item $R[g]$ for $g \in \set{1, \dots, 2k}$ is a FAID object storing a number in $\set{0, \dots, k}$
    \end{itemize}
  \end{invisiblebox}

  \begin{invisiblebox}{\KwSty{Process-local static data:}}
    \begin{itemize}\fontsize{9.5}{11.5}\selectfont
      \item $d$ is a dictionary storing keys in $\set{1, \dots, m}$ and values in $U \times \set{\HI{0,\,}1, \dots, 2k}$
    \end{itemize}
  \end{invisiblebox}

  \begin{invisiblebox}{\KwSty{Initialization:}}
    \begin{itemize}\fontsize{9.5}{11.5}\selectfont
      \item $L[i] \gets (u_i, \HI{0}^\ast)$, where $u_i$ is an application-dependent initial value from $U$
      \HI{\item $C[i] \gets 0$}
      \item $R[g] \gets 0$
      \item $d$ starts empty for each process (all indexing returns $\bot$)
    \end{itemize}
    \footnotetext{$\ast$ In the non-QHI-preserving version (black lines only), \Tag fields are initialized to any nonzero value, e.g., all ones.}
  \end{invisiblebox}

  \normalsize
  \begin{method}{LL($i \in \set{1, \dots, m}$)}
    \lIf{$d[i] \neq \bot$}{\CL{i}}\label{line:ll:reset_link}
    \HI{$C[i].\FAI{}$}\;\label{line:ll:inc_activity}
    $(val, tag) \gets L[i].\Read{}$\;\label{line:ll:first_read}
    $d[i] \gets (val, tag)$\;\label{line:ll:set_di}
    \HI{\lIf{$tag = 0$}{\Return{$val$}}}\label{line:ll:first_zero_tag}
    $R[tag].\FAI{}$\;\label{line:ll:inc_refcount}
    $(val', tag') \gets L[i].\Read{}$\;\label{line:ll:second_read}
    \lIf{$(val, tag) = (val', tag')$}{\Return{$val'$}}\label{line:ll:matching_tag}
    \HI{\If{$tag' = 0$}{
      $R[tag].\FAD{}$\;\label{line:ll:dec_refcount}
      $d[i] \gets (val', 0)$\;\label{line:ll:reset_tag}
      \Return{$val'$}\label{line:ll:second_zero_tag}
    }}
    \CL{i}\;\label{line:ll:clear}
    \Return{$val$}\label{line:ll:final_return}
  \end{method}

  \begin{method}{SC($i \in \set{1, \dots, m}, u \in U$)}
    \lIf{$d[i] = \bot$}{\Return{\False}}\label{line:sc:early_return}
    \lnl{line:sc:repeat}\Repeat{
      $R[tag'].\Read{} = 0$}
    {$tag' \gets$ \KwSty{uniform and independent random sample from} $\set{1, \dots, 2k}$\label{line:sc:choose_tag}
    }\label{line:sc:check_tag}
    $result \gets L[i]$.\CAS{d[i], (u, tag')}\;\label{line:sc:first_cas}
    \HI{\lIf{$result = \False$}{$result \gets L[i]$.\CAS{(d[i].\val, 0), (u, tag')}}}\label{line:sc:second_cas}
    \CL{i}\;\label{line:sc:clear}
    \Return{$result$}\label{line:sc:return}
  \end{method}
\end{figure}

\begin{figure}\label{fig:additional_methods}
    \begin{method}{CL($i \in \set{1, \dots, m}$)}
    \lIf{$d[i] = \bot$}{\Return}\label{line:cl:early_return}
    \nonHI{\HI{\lIf{$d[i].\Tag \neq 0$}{\nonHI{$R[d[i].\Tag].\FAD{}$}}}}\label{line:cl:dec_refcount}
    $d[i] \gets \bot$\;\label{line:cl:clear_di}
    \HI{\lIf{$C[i].\FAD{} > 1$}{\Return}\label{line:cl:dec_activity}
    $(val, tag) \gets L[i].\Read{}$\;\label{line:cl:first_read}
    \lIf{$tag = 0$}{\Return}\label{line:cl:zero_tag_return}
    $R[tag].\FAI{}$\;\label{line:cl:inc_refcount}
    \If{$L[i].\Read{} = (val, tag)$}{\label{line:cl:second_read}
      \If{$C[i].\Read{} = 0$}{\label{line:cl:confirm_quiescence}
        $L[i]$.\CAS{(val, tag), (val, 0)}\label{line:cl:cas}
      }
    }
    $R[tag].\FAD{}$\;\label{line:cl:final_return}}
  \end{method}

  \begin{method}{Read($i \in \set{1, \dots, m}$)}
    \Return{$L[i].\Read{}.\val$}\label{line:read}
  \end{method}

  \begin{method}{VL($i \in \set{1, \dots, m}$)}
    \lIf{$d[i] = \bot$}{\Return{\False}}\label{line:vl:early_return}
    \lIf{$L[i].\Read{}.\Tag \in \set{d[i].\Tag\HI{, 0}}$}{\Return{\True}}\label{line:vl:read}\label{line:vl:return_true}
    \CL{i}\;\label{line:vl:clear}
    \Return{\False}\label{line:vl:return_false}
  \end{method}
\end{figure}

\subsection{Low-Level Description}\label{sec:low-level}


We store the values of the LL/SC objects in an array $L$ of CAS objects.
In addition, each CAS object stores a non-negative bounded tag that is initialized to zero.
The array $R$ is used to count each tag's usage, so that $R[g]$ counts the number of processes that have an outstanding \LL{} operation that read the tag $g$; its purpose is to prevent ABAs.
Similarly, array $C$ is an activity counter for LL/SC objects, so that $C[i]$ tracks how many processes have an outstanding \LL{i}; its purpose is to assist in preserving history independence.

In any LL/SC implementation, processes need some way to keep track of their outstanding \LL{} operations, in order to be able to know whether their \SC{} and \VL{} calls should succeed.
Here, processes use a local dictionary $d$, which stores, for each object $i$, a pair consisting of the last value and tag read from $L[i]$, or the empty value $\bot$ if the process has no outstanding \LL{i}.
(The \SC{i, \cdot}, \VL{i}, and \CL{i} methods each return immediately when $d[i] = \bot$, since they are only meaningful for a process with an outstanding \LL{i}.)

Similarly to the high-level description, we begin by focusing on the black lines which satisfy the sequential specification without preserving history independence, and postpone discussion of the red lines, involving array $C$, most of the \CL{} method, and tag zero.

\subsubsection{LL() Method}
A calling process $p$ reads the pair stored in $L[i]$ (\autoref{line:ll:first_read}), consisting of its value and tag, and stores this state in its local dictionary $d$ (\autoref{line:ll:set_di}).
Next, $p$ increments the corresponding usage count in $R$ (\autoref{line:ll:inc_refcount}), then reads again to confirm that $L[i]$ has not changed (\autoref{line:ll:second_read}).
If $L[i]$ has not changed upon the second read, then we say that $p$ has successfully \emph{protected} the tag from being reused (see the formal definition in \autoref{def:protect}), and $p$ returns in \autoref{line:ll:matching_tag}, linearizing in \autoref{line:ll:second_read}.
This protection lasts until $p$ next invokes \CL{i}, as \autoref{line:cl:dec_refcount} decrements the usage count again.
As we later discuss in the \SC{} method, the algorithm guarantees that a tag is only reused when no process protects it.

Otherwise, $L[i]$ \emph{has} changed to a different value by the time $p$ reads it in \autoref{line:ll:second_read}.
In this case, we take advantage of the fact that an \SC{i, \cdot} is required to fail if there was an intervening \SC{i, \cdot} in between a process's \LL{i} and \SC{i, \cdot}.
We thus retroactively set $p$'s \LL{i} call's linearization point (which defines the order of the execution's linearization) to the first read in \autoref{line:ll:first_read} to satisfy the sequential specification.
(Although this is a standard trick for LL/SC implementations, this retroactive linearization point prevents this \LL{} implementation from being strongly linearizable~\cite{GHW2011a}.)
The \CL{i} call in \autoref{line:ll:clear} sets $d[i]$ to $\bot$ (\autoref{line:cl:clear_di}), as though $p$ had no outstanding \LL{i}, which forces $p$'s following \SC{i, \cdot} to fail in \autoref{line:sc:early_return}.
(In this case, $p$ still does have an outstanding \LL{i} by definition, because the fact that the next \SC{i, \cdot} will fail is not directly visible by an algorithm using our LL/SC implementation.)

\subsubsection{SC() Method}\label{sec:low-level:sc}
Line~\ref{line:sc:early_return} ensures that the \SC{} returns \False{} in the case of the calling process not having any outstanding \LL{} (lines~\ref{line:cl:early_return} and~\ref{line:vl:early_return} serve a similar purpose).
In lines~\ref{line:sc:repeat}--\ref{line:sc:check_tag}, the active process chooses a random tag in the range of possible tags, until it sees one with a usage count of zero, implying that it is not protected by any process.
Since there are $2k$ tags and at most $k$ outstanding \LL{} operations which can protect tags, each attempt succeeds with probability at least $1/2$.
Recall that $d[i]$ stores the previous value-tag pair read from $L[i]$.
Line~\ref{line:sc:first_cas} then attempts a \CAS{} from the previous value-tag pair stored in $d[i]$ to the new value (the second argument of \SC{}) and the newly chosen tag.
The method linearizes at the point of this \CAS{} attempt, calls \CL{} to clear the calling process's outstanding \LL{}, and returns whether the \CAS{} attempt succeeded.

We now highlight that since \SC{} does not itself create an outstanding \LL{}, it does not need to protect any tag, and so it does not increment any usage count in $R$.
This may appear problematic at first:
Suppose that a process $p$ pauses after seeing $R[g] = 0$ in \autoref{line:sc:check_tag} and before its \CAS{} in \autoref{line:sc:first_cas}, and during this pause another process $q$ sees that $R[g] = 0$ for the same tag $g$.
Then, process $q$ may use tag $g$ for the same or different object, and it may no longer be true that $R[g] = 0$ at the time that $p$ executes its \CAS{} in \autoref{line:sc:first_cas}.
We next explain why this does not lead to any problem.

\dbtodoi{TODO: Work on the below phrasing, since these are key ideas for our implementation.} 
First, our implementation \emph{does} allow different LL/SC objects to share the same tag.
Unlike memory management for pointers, the \emph{only} use of tags is to prevent ABAs for individual LL/SC objects in between an \LL{} and matching \SC{}.
As such, it does not matter if the same tag is used for \emph{different} LL/SC objects, as that would not cause an ABA for a single object.

Second, if process $q$ chooses the same tag $g$ and has a successful \CAS{} on the \emph{same} object as $p$, then it is possible to prove that $p$'s \CAS{} will necessarily fail, preventing an ABA by $p$.
The key \emph{protection invariant} (see \autoref{lem:invariant}(\ref{item:invariant:protection})) is that, as long as some process is protecting a tag for some object and therefore contributing to that tag's usage count, no process can change the object's \Tag field to the protected tag.
For process $p$ to have a successful \CAS{}, changing the \Tag field back to $g$ after $q$'s successful \CAS{}, the \Tag field would have had to change back to the previous tag $g'$ that $p$ was expecting from $d[i].\Tag$.
But $p$ must have already done an \LL{i} before its \SC{i, \cdot}, thus protecting that expected tag $g'$.
This is the only way that the protection invariant could be broken, but it requires that the invariant was \emph{already} broken previously for tag $g'$, and so the invariant holds by induction over the execution.


\subsubsection{Remaining Methods}
\dbtodoi{Note that the way I'm using `outstanding \LL{}' here is a little less formal than the exact definition, but I think it suffices for the intuition.
  (Perhaps an alternative definition of outstanding \LL{} could have been directly tied to a contribution to $C[i]$, but I think that's too implementation-dependent.)}
The \CL{} method decrements the usage account for $d[i].\Tag$, ending the calling process $p$'s protection of that tag, and resets $d[i]$ to clear $p$'s outstanding \LL{i}.
The \Read{} method ignores the \Tag field and simply outputs the value of the object.
The \VL{} method checks if the \Tag field has not changed, and returns the appropriate response, calling \CL{} to clear the outstanding \LL{} operation in case the \Tag field has changed.

\subsubsection{CL(), C, and Tag Zero: Preserving Quiescent History Independence}
We now discuss the role of the red lines for preserving quiescent history independence.
Recall (from page~\pageref{fig:main_methods}) that all \Tag fields are initialized to zero for the QHI-preserving version of the algorithm.
The \CL{i} method now ensures that $L[i].\Tag$ is reset to zero when it clears the last outstanding \LL{i}, so that the shared memory returns to its canonical representation upon entering a settled state.

Recall that array $C$ is a shared activity counter for outstanding \LL{} operations on each object.
As such, an \LL{i} call increments $C[i]$ in \autoref{line:ll:inc_activity}, which announces that a process has created an outstanding \LL{i}.
Line~\ref{line:cl:dec_activity} of a \CL{i} call then decrements $C[i]$; if the count before the decrement was greater than one, then some other process has an outstanding \LL{i}, so the calling process $p$ does not reset $L[i].\Tag$.
Otherwise, $p$ tries to reset $L[i].\Tag$ to zero if it is nonzero, but with several checks to avoid doing so after some other process wakes up and starts working on the object.

Lines~\ref{line:cl:first_read},~\ref{line:cl:inc_refcount}, and~\ref{line:cl:second_read} work the same way as lines~\ref{line:ll:first_read},~\ref{line:ll:inc_refcount}, and~\ref{line:ll:second_read} to try to protect the current tag $g$ stored in $L[i].\Tag$.
If $g$ is zero or if $L[i].\Tag$ changes from $g$, then \CL{i} no longer needs to reset $L[i].\Tag$, so it decrements its usage count if necessary (\autoref{line:cl:final_return}).
Otherwise, after double-checking that no process has created an outstanding \LL{} on object $i$ while $p$ was protecting $g$ (\autoref{line:cl:confirm_quiescence}), $p$ resets the \Tag field via a \CAS{} in \autoref{line:cl:cas}.
This last check ensures that any process which changes $L[i].\Tag$ does so after $p$'s execution of \autoref{line:cl:confirm_quiescence}, while $g$ is protected.
Similarly to \SC{}, if the \Tag field changes after $p$ reads it a second time in \autoref{line:cl:second_read}, then the protection invariant ensures that the \CAS{} in \autoref{line:cl:cas} fails, because $p$ is protecting the tag that it read in \autoref{line:cl:first_read}.

\dbtodo{I removed the term `AB0 transition' term and rephrased this explanation with `strong protection', although in hindsight that does not feel like a very self-descriptive term.}
This process of resetting a \Tag field to zero does \emph{not} change the abstract state of the object, but it \emph{is} possible for another process to call \LL{i} on the object in between lines~\ref{line:cl:confirm_quiescence} and~\ref{line:cl:cas}, creating a newly outstanding \LL{i}.
Thus, $L[i].\Tag$ can still change from an \LL{i} operation's newly protected tag $g$ to zero.
However, if $L[i].\Tag$ changes to some other tag $g' \notin \set{g, 0}$ while $g$ is protected, then it cannot change to zero again, because $C[i] > 0$ while there is an outstanding \LL{i} (from \autoref{line:ll:inc_activity}).
As such, once a process has an outstanding \LL{i} with tag $g$ and later sees $L[i].\Tag = 0$, it knows that $L[i].\Tag$ changed \emph{directly} from $g$ to zero due to a \CL{i} call, without any successful \SC{i, \cdot} linearizing in between.
We refer to this additional guarantee as \emph{strongly protecting} a tag (see \autoref{lem:invariant}(\ref{item:invariant:strong_protection})), and it simplifies the way that other methods treat tag zero.

If an \LL{i} operation reads tag zero from either of its reads in line~\ref{line:ll:first_read} or~\ref{line:ll:second_read}, then it stores tag zero in $d[i].\Tag$, and the operation returns in line~\ref{line:ll:first_zero_tag} or~\ref{line:ll:second_zero_tag}.
Since tag zero can only be set by \CL{i} after confirming that $C[i] = 0$, it does not need a usage count in $R$, so the \LL{} operation simply linearizes as soon as it reads tag zero.
In the \SC{} method, the calling process $p$ attempts two \CAS{} operations, in lines~\ref{line:sc:first_cas} and~\ref{line:sc:second_cas}, because the tag could change from a nonzero value directly to zero.
Since $p$ strongly protects the tag read in its preceding \LL{} call, this is the only possible transition during the \SC{} call which does not change the abstract state of the object.
Thus, the two \CAS{} operations are sufficient, and the \SC{} operation linearizes at its last \CAS{} attempt.
Lastly, the \VL{} method also returns \True in the case of tag zero in \autoref{line:vl:return_true}, which again works due to the calling process strongly protecting the tag read in its preceding \LL{} call.


\fullonly{
\section{\confonly{Partial }Analysis}\label{sec:analysis}

\confonly{
  In this appendix, we provide an abbreviated version of the algorithm's analysis, with simplified statements and fewer proofs.
  Please refer to the full version~\cite{BHW2026a-full} for the complete analysis.
}

\subsection{Complexity}

We begin the analysis by showing, in \autoref{thm:complexity}, that the expected step complexity of the implementation is constant under the weak adaptive adversary.
We prove complexity first \fullonly{because}\confonly{for consistency with the full version, in which} the statements used to prove \autoref{thm:complexity} are reused later to prove history independence (\autoref{thm:history_independence}) and linearizability (\autoref{thm:linearizable}).

\fullonly{
  \begin{lemma}\label{lem:cl_clears_di}
    When a process responds from a \CL{i} or \SC{i, \cdot} call, its value of $d[i]$ is $\bot$.
  \end{lemma}
  \begin{proof}
    For \CL{i}, if the calling process $p$ returns in \autoref{line:cl:early_return}, then the lemma is satisfied by the if-condition in that line.
    Otherwise, $p$ sets $d[i]$ to $\bot$ in \autoref{line:cl:clear_di} and does not change it again until its response.

    For \SC{i, \cdot}, if the calling process $p$ returns in \autoref{line:sc:early_return}, then the lemma is satisfied by the if-condition in that line.
    Otherwise, $p$ calls \CL{i} in \autoref{line:sc:clear}, which by the above ensures $d[i] = \bot$ by its response, and then immediately returns.
  \end{proof}
}

\fullonly{
  \begin{lemma}\label{lem:nonbot_outstanding_ll}
    At point $t$, if $d[i] \neq \bot$ for process $p$, then $p$ has an outstanding \LL{i}.
  \end{lemma} 
  \begin{proof}
    Only \LL{i} can set $d[i]$ to a non-$\bot$ value.
    By \autoref{lem:cl_clears_di}, the \CL{i} and \SC{i, \cdot} methods ensure $d[i] = \bot$ by their respective responses.
    If a \VL{i} call returns \False, then either $d[i] = \bot$ in \autoref{line:vl:early_return}, or it calls \CL{i} in \autoref{line:vl:clear}, so in both cases $d[i] = \bot$ by the point it returns \False.
    Thus, if $d[i] \neq \bot$ while $p$ is idle, then $p$ has called \LL{i} and not subsequently called \SC{i, \cdot}, \CL{i}, or a \VL{i} call which returns \False.
    Hence, $p$ has an outstanding \LL{i}.
    If $d[i] \neq \bot$ while $p$ has a pending method call, then we also have that $p$ has invoked \LL{i} and not yet responded from a subsequent \SC{i, \cdot} or \CL{i} call, or a \VL{i} call which returns \False.
    Thus, also in this case $p$ has an outstanding \LL{i}.
  \end{proof}
}

\begin{definition}[Contribution]\label{def:contribute}
  Given a point $t$, a process $p$, and an element of $C$ or $R$, $p$ \emph{contributes $j$} to that element's value, if $j$ is the number of \FAI{} minus \FAD{} operations that $p$ has executed to that element before $t$.
\end{definition}

\fullonly{We next prove upper limits on $p$'s contribution to $C$ and $R$, given that $k$ is the upper bound on the total number of outstanding \LL{} operations across all processes.}

\begin{lemma}\label{lem:max_contribution}
  \fullonly{
    When a process $p$ invokes or responds from a method call,
    \begin{enumerate}[(a)]
      \item if $d[i] = \bot$, then $p$ contributes zero to $C[i]$, and otherwise $p$ contributes one to $C[i]$, and\label{item:idle_contribution:C}
      \item the total contribution to $R$ by $p$ equals the number of values in $p$'s dictionary $d$ which have nonzero tags.\label{item:idle_contribution:R}
    \end{enumerate}
    While process $p$ has a pending method call,
    \begin{enumerate}[(a)]
      \setcounter{enumi}{2}
      \item process $p$ contributes at most one to each individual element of $C$, and\label{item:contribution:C}
      \item $p$'s total contribution to $R$ is no larger than the maximum number of non-$\bot$ values in $p$'s dictionary $d$ in the interval starting immediately before and ending immediately after the method call.\label{item:p_contribution:R}
    \end{enumerate}
    At all points in the execution,
    \begin{enumerate}[(a)]
      \setcounter{enumi}{4}
      \item the total contribution from all processes to elements of $R$ is at most $k$.\label{item:tot_contribution:R}
    \end{enumerate}
  }
  \confonly{At all points in the execution, the total contribution from all processes to elements of $R$ is at most $k$.}
\end{lemma}
\fullonly{
  \begin{proof}
    Note that by \autoref{def:contribute}, $p$'s contributions to $R$ and $C$, as well as its dictionary $d$, do not change while $p$ is idle.
    Thus, if~(\ref{item:idle_contribution:C}) to~(\ref{item:p_contribution:R}) hold when $p$ responds from a method call and becomes idle, then they also hold until the point when $p$ next invokes a method.

    From~\autoref{lem:nonbot_outstanding_ll}, $p$'s number of outstanding \LL{} operations is at least as large as the number of non-$\bot$ values in $p$'s dictionary $d$.
    In the case of a pending method call, $p$'s number of outstanding \LL{} operations is at least as large as the maximum number of non-$\bot$ values in $d$ in the interval starting immediately before and ending immediately after the method call.
    Therefore,~(\ref{item:idle_contribution:R}) and~(\ref{item:p_contribution:R}) imply that $p$'s total contribution to $R$ never exceeds its number of outstanding \LL{} operations.
    Together with the global bound of $k$ on outstanding \LL{} operations,~(\ref{item:tot_contribution:R}) follows.

    As such, it suffices to prove~(\ref{item:idle_contribution:C}) and~(\ref{item:idle_contribution:R}) for a process $p$ by induction over its method calls, while proving that~(\ref{item:p_contribution:R}) and~(\ref{item:contribution:C}) hold at all points within each of $p$'s method calls.

    Initially, $p$ contributes zero to elements of $C$ and $R$, and its dictionary $d$ is empty (all values are $\bot$), so~(\ref{item:idle_contribution:C}) to~(\ref{item:p_contribution:R}) hold before $p$'s first method call.

    Suppose that~(\ref{item:idle_contribution:C}) to~(\ref{item:p_contribution:R}) hold for $p$ immediately before the invocation of a given method: we will prove that~(\ref{item:p_contribution:R}) and~(\ref{item:contribution:C}) hold throughout the method call, and that~(\ref{item:idle_contribution:C}) and~(\ref{item:idle_contribution:R}) hold immediately after the call's response.

    If the method call is \Read{i}, then~(\ref{item:idle_contribution:C}) to~(\ref{item:p_contribution:R}) hold because \Read{} does not affect $p$'s contributions or its dictionary $d$.

    We next consider the \CL{i} method.
    If the call returns immediately in \autoref{line:cl:early_return}, then neither $p$'s contributions nor $d$ change, so~(\ref{item:idle_contribution:C}) to~(\ref{item:p_contribution:R}) remain true.

    Otherwise, $d[i] \neq \bot$ at \CL{i}'s invocation, so by~(\ref{item:idle_contribution:C}), $p$'s contribution to $C[i]$ is one at invocation.
    Then, $p$ resets $d[i]$ to $\bot$ in \autoref{line:cl:clear_di} and decrements $C[i]$ in \autoref{line:cl:dec_activity}, making its contribution to $C[i]$ become zero and \textbf{satisfying~(\ref{item:idle_contribution:C})} at all the remaining response points.
    If $d[i].\Tag \neq 0$ at invocation, then $p$ decrements $R[d[i].\Tag]$ in \autoref{line:cl:dec_refcount}, satisfying~(\ref{item:idle_contribution:R}) after \autoref{line:cl:clear_di} and before \autoref{line:cl:inc_refcount}, in particular including $p$'s possible response in either line~\ref{line:cl:dec_activity} or~\ref{line:cl:zero_tag_return}.
    Since this possible decrement to $R[d[i].\Tag]$ occurs only if $p$ sets $d[i]$ to $\bot$ in the call,~(\ref{item:p_contribution:R}) remains true from invocation until at least $p$ executes \autoref{line:cl:inc_refcount}.
    Since $p$ does not increase its contribution to $C$ throughout the method call,~\textbf{(\ref{item:contribution:C}) remains satisfied} throughout the call.

    If $p$ responds in \autoref{line:cl:final_return}, then it increments and later decrements some element of $R$ (in lines~\ref{line:cl:inc_refcount} and~\ref{line:cl:final_return}).
    Since the two counteract each other's effect, and since~(\ref{item:idle_contribution:R}) holds immediately before \autoref{line:cl:inc_refcount},~\textbf{(\ref{item:idle_contribution:R}) holds} immediately after $p$'s response in \autoref{line:cl:final_return}, in addition to the other possible response points as proven above.
    Since~(\ref{item:idle_contribution:R}) implies~(\ref{item:p_contribution:R}), it also follows that~(\ref{item:p_contribution:R}) holds immediately after $p$'s response, and it only remains to show that~(\ref{item:p_contribution:R}) holds while $p$ has a pending next step, in particular after executing \autoref{line:cl:inc_refcount}.
    If $d[i].\Tag \neq 0$ at invocation, then $p$'s contribution to $R$ decrements, increments, then decrements again in lines~\ref{line:cl:dec_refcount},~\ref{line:cl:inc_refcount}, and~\ref{line:cl:final_return}, so~(\ref{item:p_contribution:R}) holds throughout the call.
    Otherwise, $d[i].\Tag = 0$ at invocation, and the number of values in $p$'s dictionary which have nonzero tags at invocation is strictly smaller than the number of values which are non-$\bot$ at invocation.
    There is therefore room for at least one increment of $p$'s contribution to $R$ while the method call is pending, and therefore the temporary increment in \autoref{line:cl:inc_refcount} still satisfies~(\ref{item:p_contribution:R}).
    Thus,~\textbf{(\ref{item:p_contribution:R}) holds} throughout a \CL{i} call.


    Since the \SC{} and \VL{} methods do not change $p$'s contribution to $C$ or $R$ or its dictionary $d$ except via their internal \CL{} calls,~(\ref{item:idle_contribution:C}) to~(\ref{item:p_contribution:R}) hold for \SC{} and \VL{} by the above arguments for \CL{}.

    It remains to show~(\ref{item:idle_contribution:C}) to~(\ref{item:p_contribution:R}) for an \LL{i} call by $p$.
    Since they each remain true after internal \CL{} calls by the above arguments, they each remain true immediately before the \LL{i} call executes \autoref{line:ll:inc_activity}.
    By \autoref{lem:cl_clears_di}, the \CL{i} call in \autoref{line:ll:reset_link} also ensures that $d[i] = \bot$ immediately before \autoref{line:ll:inc_activity}.

    It follows from~(\ref{item:idle_contribution:C}) that $p$'s contribution to $C[i]$ immediately before \autoref{line:ll:inc_activity} is zero, and hence its contribution immediately after the \FAI{} is one.
    In particular, at no point is $p$'s contribution to $C[i]$ greater than one, \textbf{satisfying~(\ref{item:contribution:C})} throughout the call.
    After executing \autoref{line:ll:set_di},~(\ref{item:idle_contribution:C}) becomes true again, and remains true until $p$'s possible invocation of \CL{i} in \autoref{line:ll:clear}, then remains true at its response by the above arguments for \CL{}, and therefore~\textbf{(\ref{item:idle_contribution:C}) is true} at any possible response from \LL{i}.
    (Note that \autoref{line:ll:reset_tag}, if executed, only changes the value of $d[i]$ from one non-$\bot$ value to another.)

    We next verify that~\textbf{(\ref{item:idle_contribution:R}) holds} at each relevant point, again using that~(\ref{item:idle_contribution:C}) to~(\ref{item:p_contribution:R}) and $d[i] = \bot$ hold immediately before \autoref{line:ll:inc_activity}.
    Since only $d[i]$ is modified in the method, it suffices to show that $d[i].\Tag \neq 0$ at the response if and only if $p$ increments its contribution to $R$ by one by its response.
    At a response in \autoref{line:ll:first_zero_tag}, $d[i].\Tag = 0$ and $R$ is untouched by $p$.
    At a response in \autoref{line:ll:matching_tag} and immediately before calling \CL{i} in \autoref{line:ll:clear}, $d[i].\Tag \neq 0$ and one element of $R$ is incremented.
    At a response in \autoref{line:ll:second_zero_tag}, $d[i].\Tag = 0$ and the increment of $R$ is reversed.
    Finally, since~(\ref{item:p_contribution:R}) and $d[i] = \bot$ hold immediately after \autoref{line:ll:reset_link}, it follows that setting $d[i]$ to a non-$\bot$ value in \autoref{line:ll:set_di} allows for one increment to $p$'s contribution to $R$ while still satisfying~(\ref{item:p_contribution:R}).
    Since $p$ can increment its contribution to $R$ only in \autoref{line:ll:inc_refcount} of the method,~\textbf{(\ref{item:p_contribution:R}) remains satisfied}, even if $p$ later reverses this increment.
  \end{proof}
}

\begin{lemma}\label{lem:more_inc_than_dec}
  \fullonly{
    Let $p$ be a process and $t$ a point in the execution.
    For each $i \in \set{1, \dots, m}$, the number of \FAI{} operations on $C[i]$ by $p$ before $t$ is greater than or equal to the number of \FAD{} operations on $C[i]$ by $p$ before $t$.
    Similarly, for each $g \in \set{1, \dots, 2k}$, the number of \FAI{} operations on $R[g]$ by $p$ before $t$ is greater than or equal to the number of \FAD{} operations on $R[g]$ by $p$ before $t$.
  }
  \confonly{Each process contributes a non-negative value to each element of $C$ and $R$.}
\end{lemma}
\fullonly{
  \begin{proof}
    To prove the lemma, it suffices to show that, for each \FAD{} operation on an element of $C$ or $R$, there is a one-to-one mapping to a previous \FAI{} operation by the same process on the same element.
    We first define the mapping, and then show that it is one-to-one.

    A \FAD{} on $R$ in \autoref{line:ll:dec_refcount} maps to the process's previous execution of \autoref{line:ll:inc_refcount} in the same \LL{} method, both of which affect the same element of $R$.
    Similarly, a \FAD{} on $R$ in \autoref{line:cl:final_return} maps to the process's previous execution of \autoref{line:cl:inc_refcount} in the same \CL{} method, both of which affect the same element of $R$.

    If a \CL{i} call invoked by process $p$ proceeds past \autoref{line:cl:early_return}, then $p$ confirms that $d[i] \neq \bot$ in that line.
    Since the dictionary $d$ starts empty and only \LL{i} can set $d[i]$ to a non-$\bot$ value, there is some last \LL{i} call invoked by $p$ before the \CL{i} call which sets $d[i]$ to that non-$\bot$ value.
    Both of the remaining \FAD{} operations are in the \CL{i} method, so we map the decrements to the corresponding increments in this \LL{i} call.
    Specifically, \autoref{line:cl:dec_refcount} in \CL{i} corresponds to \autoref{line:ll:inc_refcount} in the corresponding \LL{i} call (which $p$ executes since $d[i].\Tag \neq 0$ when $p$ invokes \CL{i}), and \autoref{line:cl:dec_activity} in \CL{i} maps to \autoref{line:ll:inc_activity} of the corresponding \LL{i} call.
    In both cases, the \FAD{} and \FAI{} operations are by the same process on the same object.
    For lines~\ref{line:cl:dec_refcount} and~\ref{line:ll:inc_refcount} in particular, $d[i].\Tag$ stays the same because only \LL{i} can set $d[i]$ to a non-$\bot$ value, $d[i] \neq \bot$ at \CL{i}'s invocation, and the mapped execution of \autoref{line:ll:inc_refcount} is from the most recent \LL{i} call by $p$.

    We now show that each \FAI{} is mapped to by at most one \FAD{} on the same object (i.e., the mapping is one-to-one).
    This is immediate for the mapping from \autoref{line:cl:final_return} to \autoref{line:cl:inc_refcount}, since they involve the same method call and no other \FAD{} maps to \autoref{line:cl:inc_refcount}.

    Next, consider the \FAI{} on $R$ in \autoref{line:ll:inc_refcount} in some \LL{i} call by process $p$.
    If $p$ executes \autoref{line:ll:dec_refcount} in the same call, then it also sets $d[i].\Tag$ to zero in \autoref{line:ll:reset_tag} before returning from the call.
    Therefore, $p$ does not execute the \FAD{} in \autoref{line:cl:dec_refcount} of \CL{i} until after it invokes another \LL{i} which sets $d[i]$ to a non-$\bot$ value (and specifically sets $d[i].\Tag$ to a nonzero value).
    If $p$ does execute a \FAD{} in \autoref{line:cl:dec_refcount} which maps to the \FAI{} in \autoref{line:ll:inc_refcount}, then it resets $d[i]$ to $\bot$ in \autoref{line:cl:clear_di} before returning.
    Therefore, if $p$ invokes \CL{i} again before invoking another \LL{i} call which sets $d[i]$ to a non-$\bot$ value, $p$ does not get past \autoref{line:cl:early_return}, and in particular does not execute \autoref{line:cl:dec_refcount}.
    Thus, the \FAI{} operation in \autoref{line:ll:inc_refcount} can only be mapped to by at most one \FAD{} operation.

    Finally, consider the \FAI{} operation on $C$ in \autoref{line:ll:inc_activity} in some \LL{i} call by $p$, which can be mapped to by a \FAD{} operation in \autoref{line:cl:dec_activity} in a \CL{i} call by $p$.
    Before executing \autoref{line:cl:dec_activity}, $p$ resets $d[i]$ to $\bot$ in \autoref{line:cl:clear_di}, so it cannot get past \autoref{line:cl:early_return} again until invoking \LL{i} and setting $d[i]$ to a non-$\bot$ value again.
  \end{proof}
}

\begin{lemma}\label{lem:value_equals_contributions}
  The value of each element in $C$ and $R$ equals the sum of the contributions from each process (i.e., no element overflows).
\end{lemma}
\fullonly{
  \begin{proof}
    From \autoref{def:contribute} and \autoref{lem:more_inc_than_dec}, we have that each process contributes a non-negative value to each element of $C$ and $R$.
    We also have that only \FAI{} and \FAD{} operations can affect that value of elements of $C$ or $R$.
    The lemma then follows from \autoref{lem:max_contribution}(\ref{item:contribution:C}) and~(\ref{item:tot_contribution:R}), since the $n$ processes can contribute at most $n$ to the value of each element of $C$, and at most $k$ in total to elements of $R$, which are their respective maximum values.
  \end{proof}
}

\begin{theorem}[Expected Complexity]\label{thm:complexity}
  Each method other than \SC{} has constant step complexity.
  The expected step complexity of \SC{} is constant under a weak adaptive adversary.
\end{theorem}
\begin{proof}
  The first part is immediate from the fact that the only loop in the entire implementation is \linerefrange{line:sc:repeat}{line:sc:check_tag} of the \SC{} method.

  For the second part, recall that under the weak adaptive adversary, a random choice by process $p$ is immediately followed by its next shared memory step.
  Thus, when $p$ chooses a random tag $g$ in \autoref{line:sc:choose_tag}, it then immediately reads $R[g]$ in \autoref{line:sc:check_tag} without any other process taking a step in between the two lines.
  In particular, the value of $R[g]$ does not change in between the two lines, and it stays at zero if it was zero at $p$'s execution of \autoref{line:sc:choose_tag}.
  Thus, $p$ breaks out of its loop as soon as it chooses a tag $g$ such that $R[g] = 0$ in \autoref{line:sc:choose_tag}.
  Since this random choice is uniform over $\set{1, \dots, 2k}$ and independent, it suffices to show that, at any given time, at least a constant fraction of the indices in $R$ have a value of zero.

  From Lemmas~\ref{lem:more_inc_than_dec} and~\ref{lem:value_equals_contributions}, a given element of $R$ is nonzero at time $t$ exactly if there exists some process that contributes to that element at time $t$.
  From \autoref{lem:max_contribution}\fullonly{(\ref{item:tot_contribution:R})}, the total contribution from all process to all elements of $R$ is at most $k$.
  Thus, at least $k$ of $R$'s $2k$ elements have a value of zero at any given time, which completes the proof.
\end{proof}

\subsection{History Independence}
The goal of this section is to prove \autoref{thm:history_independence}, which shows that when the system is in a settled state, then the shared memory depends only on the LL/SC objects' values, therefore achieving the QHI-preserving property discussed in \autoref{sec:qhi-preserving}.
In our case, we will prove that in a settled state, $L[i].\val$ stores the value of LL/SC object $i$, and all other shared memory used in the implementation equals zero, like in the initial state.

\fullonly{
  \begin{lemma}\label{lem:incremented_nonzero} 
    If a process has performed a \FAI{} operation on a given element of $C$ or $R$ before point $t$, and has not performed a subsequent \FAD{} operation before point $t$, then that element's value is nonzero at point $t$.
  \end{lemma}
  \begin{proof}
    From \autoref{lem:more_inc_than_dec}, for each process and the given element $e$ of $C$ or $R$, the process's contribution to $e$ is non-negative at all points.
    Thus, immediately after process $p$ increments $e$, $p$'s contribution as well as the total contribution to $e$ becomes positive, and this remains true until $p$ next decrements $e$.
    The lemma then follows from \autoref{lem:value_equals_contributions}.
  \end{proof}
}

\begin{definition}[Protected and Strongly Protected Tags]\label{def:protect}
  A tag $g$ becomes \emph{protected} for LL/SC object $i$ at point $t$ by a process $p$ if either
  \begin{enumerate}[(a)]
    \item $p$ reads zero from $L[i].\Tag$ in \autoref{line:ll:first_read} or~\ref{line:ll:second_read} at point $t$, in which case $g = 0$, or\label{item:zero_protected}
    \item $p$ reads $g$ from $L[i].\Tag$ in \autoref{line:ll:second_read} or~\ref{line:cl:second_read} at point $t$, such that its local variable pair $(val, tag)$ (read from $L[i]$ in \autoref{line:ll:first_read} or~\ref{line:cl:first_read} respectively) equals the value of $L[i]$ at point $t$.\label{item:nonzero_protected}
  \end{enumerate}
  If a tag becomes protected for object $i$ during an \LL{i} call, then it is also \emph{strongly protected} for $i$ by $p$.
  In that case, the protection and strong protection by process $p$ end when it next invokes \CL{i}.
  Otherwise, the tag becomes protected for object $i$ by $p$ during a \CL{i} call, in which case it is not strongly protected by $p$, and the protection ends when that call responds.

  \fullonly{If a tag is protected or strongly protected by any process at point $t$ for object $i$, we also simply refer to it as protected or strongly protected for object $i$, without referring to the protecting process.}
\end{definition}
We will later show in \autoref{lem:invariant} that a tag protected for $i$ prevents ABAs on $L[i]$, and a tag strongly protected for $i$ prevents $L[i].\Tag$ from returning to zero from a different tag.
Note that although distinct LL/SC objects' \Tag fields can store the same tag, protection and strong protection always refer to a specific object.

\begin{lemma}\label{lem:protected_activity}
  If some tag is strongly protected for object $i$ at point $t$, then $C[i] \neq 0$ at point $t$.
\end{lemma}
\fullonly{
  \begin{proof}
    From \autoref{def:protect}, a tag can only become strongly protected for $i$ by a process $p$ in an \LL{i} call in \autoref{line:ll:first_read} or~\ref{line:ll:second_read}, and both lines happen after $p$ increments $C[i]$ in \autoref{line:ll:inc_activity} and before its next invocation of \CL{i}.
    Since only \CL{i} can decrement $C[i]$ (in \autoref{line:cl:dec_activity}), it follows that $p$'s strong protection of $g$ for $i$ ends before it decrements $C[i]$.
    The lemma then follows from \autoref{lem:incremented_nonzero}.
  \end{proof}
}

\begin{lemma}\label{lem:idle_protection}
  If a process $p$ has no pending method call on object $i$ at point $t$, then
  \begin{enumerate}[(a)]
    \item $p$ strongly protects tag $g$ for object $i$ if and only if $p$'s value of $d[i].\Tag = g$ at point $t$, and\label{item:di_iff_protected}
    \item if $p$ does strongly protect $g$, then $p$'s value of $d[i]$ at point $t$ equals the value of $L[i]$ when this strong protection started.\label{item:di_Li_match_when_protected}
  \end{enumerate}
\end{lemma}
\fullonly{
  Note that \autoref{lem:idle_protection} also implies that $p$ can only protect at most one tag per object while idle.
  \begin{proof}
    We first note that the only case in which $p$ can protect a tag for $i$ while not having a pending method call on object $i$, by \autoref{def:protect}, is if it protected the tag as part of an \LL{i} call.
    In particular, this implies that $p$ also strongly protects the tag.
    Moreover, only \LL{i} can set $d[i]$ to a non-$\bot$ value (in lines~\ref{line:ll:set_di} and~\ref{line:ll:reset_tag}) and only \CL{i} can reset $d[i]$ to $\bot$ (in \autoref{line:cl:clear_di}).
    An invocation of \CL{i} ends $p$'s (strong) protection of any tag for object $i$ by \autoref{def:protect}, and, from \autoref{lem:cl_clears_di}, $d[i] = \bot$ when \CL{i} responds.
    To prove~(\ref{item:di_iff_protected}), it remains to show that an \LL{i} call strongly protects tag $g$ for $i$ by its response if and only if $d[i].\Tag = g$ at that call's response.
    To prove~(\ref{item:di_Li_match_when_protected}), we need to also show that $d[i]$ is set to the value read from $L[i]$ at the point that $g$ become strongly protected.

    To show these, we consider the five possible response points from an \LL{i} call by process $p$.
    Suppose by induction that the lemma is true before the call begins (the base case holds since $d[i] = \bot$ and processes protect no tags at $t = 0$).
    If $p$ protects some tag $g$ for $i$ just before invocation of the \LL{i} call, then $d[i].\Tag = g$ before invocation, and therefore $p$ invokes \CL{i} in \autoref{line:ll:reset_link}.
    Thus, after \autoref{line:ll:reset_link}, $p$ does not protect any tag for $i$, and $d[i] = \bot$.

    If $p$ responds in \autoref{line:ll:first_zero_tag} or~\ref{line:ll:second_zero_tag}, then $p$ reads zero from $L[i].\Tag$ in \autoref{line:ll:first_read} or~\ref{line:ll:second_read} respectively; at the point of that read, by \autoref{def:protect}(\ref{item:zero_protected}), $p$ strongly protects tag zero for $i$.
    In \autoref{line:ll:set_di} or~\ref{line:ll:reset_tag} respectively, $p$ sets $d[i]$ to the value read from $L[i]$ at the point that it started strongly protected tag zero, then it responds without further changing $d[i]$.

    If $p$ responds in \autoref{line:ll:matching_tag}, then $p$ re-reads its local variables $(val, tag)$ from $L[i]$ in \autoref{line:ll:second_read}, after having set $d[i]$ to $(val, tag)$ in \autoref{line:ll:set_di} and not changing it afterwards.
    By \autoref{def:protect}(\ref{item:nonzero_protected}), $p$ strongly protects $d[i].\Tag$ in \autoref{line:ll:second_read}, at which point $L[i] = d[i]$.

    Otherwise, $p$ responds in \autoref{line:ll:final_return}, it which case the \CL{i} call invoked in \autoref{line:ll:clear} resets $d[i]$ to $\bot$ and ends any protection of a tag by $p$ for $i$.
  \end{proof}
}

\begin{lemma}\label{lem:protect_during_sc}
  Suppose that a successful \CAS{} takes place on $L[i]$ at point $t$, such that $L[i].\Tag = g_1$ immediately before $t$ and $L[i].\Tag = g_2 \neq 0$ immediately after $t$.
  At point $t$, some process $p$ executes \autoref{line:sc:first_cas} or~\ref{line:sc:second_cas} of an \SC{i, \cdot} call.
  Process $p$ strongly protects some tag $g^\star$ for $i$ from the invocation of its call until at least point $t$, and $g^\star = g_1$ if $g_1 \neq 0$.
\end{lemma}
\fullonly{
  \begin{proof}
    Since only lines~\ref{line:sc:first_cas} and~\ref{line:sc:second_cas} can change $L[i].\Tag$ to a nonzero value, it follows that point $t$ is the point that some process $p$ executes one of the two lines.
    Since $p$ does not return in \autoref{line:sc:early_return}, by \autoref{lem:idle_protection}(\ref{item:di_iff_protected}), $p$ strongly protects $g^\star = d[i].\Tag$ for $i$ when it invokes \SC{i, \cdot}.
    Moreover, since $p$ does not invoke \CL{i} in between invoking \SC{i, \cdot} and its successful \CAS{} at point $t$, by \autoref{def:protect} $g^\star$ remains strongly protected for $i$ until at least point $t$.
    Finally, if $g_1 \neq 0$, then point $t$ can only be an execution of \autoref{line:sc:first_cas}, in which case $g^\star = g_1$.
  \end{proof}
}

\begin{lemma}\label{lem:unprotected}
  If at point $t$ no process has a pending method call on object $i$ and no process strongly protects a tag for $i$, then $C[i] = 0$.
  If at point $t$ no process has a pending method call on any object and no process strongly protects any tag, then $R[g] = 0$ for all $g \in \set{1, \dots, 2k}$.
\end{lemma}
\fullonly{
  \begin{proof}
    Suppose that no process has a pending method call on object $i$ and no process strongly protects any tag for $i$ at point $t$.
    If $d[i] \neq \bot$ for some process, then $d[i].\Tag = g$ for some tag $g \in \set{0, \dots, 2k}$, and from \autoref{lem:idle_protection}(\ref{item:di_iff_protected}) $p$ would strongly protect $g$ at point $t$: thus, $d[i] = \bot$ for all processes at point $t$.
    \autoref{lem:max_contribution}(\ref{item:idle_contribution:C}) and \autoref{lem:value_equals_contributions} then imply that $C[i] = 0$ at point $t$.
    If no process has any pending method call or strongly protects any tag at point $t$, then \autoref{lem:max_contribution}(\ref{item:idle_contribution:R}) and \autoref{lem:value_equals_contributions} imply that $R[g] = 0$ for all $g$.
  \end{proof}
}

\begin{lemma}\label{lem:quiescent_zero_tag}
  If at point $t$ no process has a pending method call on object $i$ and no process strongly protects a tag for $i$, then $L[i].\Tag = 0$ at point $t$.
\end{lemma}
\begin{proof}
  The lemma is true for $t = 0$, since $L[i].\Tag$ is initialized to zero for each $i \in \set{1, \dots, m}$.
  The \Read{i} method does not affect which tags the calling process protects or the value of $L[i].\Tag$, so we need not consider it for this lemma.
  Similarly, by \autoref{lem:idle_protection}(\ref{item:di_iff_protected}), we have that method calls which respond in either \autoref{line:sc:early_return},~\ref{line:cl:early_return}, or~\ref{line:vl:early_return} do not affect which tags the calling process strongly protects or the value of $L[i].\Tag$.
  Method calls on objects $j \neq i$ also do not affect whether the calling process strongly protects a tag for $i$ of the value of $L[i].\Tag$.
  Note that all the method calls mentioned so far are either not for object $i$, or consist of at most one shared memory step.
  Therefore, none of them change whether a process has a pending method call on object $i$.
  We can therefore assume without loss of generality that point $t$ is the response point of a method not discussed so far.%
  \fullonly{

    By inspection of the code and by \autoref{def:protect} and \autoref{lem:idle_protection}(\ref{item:di_iff_protected}), the remaining method calls for which the calling process responds without strongly protecting a tag for $i$ at the response point are exactly
    \begin{enumerate}[(a)]
      \item an \LL{i} call that responds in \autoref{line:ll:final_return},
      \item an \SC{i, \cdot} call that responds in \autoref{line:sc:return},
      \item a \CL{i} call that responds after \autoref{line:cl:early_return}, and
      \item a \VL{i} call that responds in \autoref{line:vl:return_false}.
    \end{enumerate}
    Since the%
  }
  \confonly{Since the remaining }cases of responding in either \autoref{line:ll:final_return},~\ref{line:sc:return}, or~\ref{line:vl:return_false} all immediately follow a \CL{i} call, we can further assume that point $t$ is the response point of a \CL{i} call $c$.

  By \autoref{lem:unprotected}, $C[i] = 0$ at point $t$.
  Since call $c$ does not respond in \autoref{line:cl:early_return}, it decrements $C[i]$ in \autoref{line:cl:dec_activity}.
  Therefore, there is a point $t'$ at which $C[i]$ is last decremented to zero before point $t$.
  (Since $t$ is the response of a \CL{i} call, and a \CL{i} call does not respond in \autoref{line:cl:dec_activity} if it decrements $C[i]$ to zero, it follows that $t$ is not itself a point in which $C[i]$ is decremented to zero.)
  Since only \autoref{line:cl:dec_activity} decrements $C[i]$, it follows that some \CL{i} call $c'$ (possibly $c' = c$) decrements $C[i]$ to zero at point $t'$.
  By definition of $t'$, $C[i] = 0$ throughout $(t', t]$, and so by \autoref{lem:protected_activity} we have that no tag is strongly protected for $i$ at any point in $(t', t]$.
  From \autoref{lem:protect_during_sc}, no successful \CAS{} takes place on $L[i]$ in $(t', t]$ that sets $L[i].\Tag$ to a nonzero value.
  It therefore suffices to show that $L[i].\Tag = 0$ at some point in $(t', t]$.

  Let $p$ be the process executing $c'$.
  Since no process has a pending method call on object $i$ at point $t$, it follows that $c'$ responds no later than point $t$.
  If $c'$ responds in \autoref{line:cl:zero_tag_return}, then $L[i].\Tag = 0$ when $p$ executes \autoref{line:cl:first_read} after point $t'$.
  Otherwise, since $L[i].\Tag$ does not change to another nonzero value in $(t', t]$, when $p$ reads $L[i]$ for the second time in \autoref{line:cl:second_read}, either $L[i].\Tag = 0$, or $L[i]$ is unchanged.
  If $L[i]$ is unchanged in \autoref{line:cl:second_read}, then since $C[i] = 0$ throughout $(t', t]$, $p$ executes \autoref{line:cl:cas}.
  If $p$ does execute \autoref{line:cl:cas}, then either the \CAS{} successfully changes $L[i].\Tag$ to zero, or it fails because $L[i].\Tag$ is already zero.
\end{proof}

\confonly{
  \begin{lemma}\label{lem:nonbot_outstanding_ll}
    At point $t$, if $d[i] \neq \bot$ for process $p$, then $p$ has an outstanding \LL{i}.
  \end{lemma} 
}

\begin{theorem}[QHI-Preserving Property]\label{thm:history_independence}
  If every process is idle and has no outstanding \LL{} operation at point $t$, then all values of $C$, $R$, and all \Tag fields of $L$ are zero at point $t$.
\end{theorem}
\begin{proof}
  From \autoref{lem:nonbot_outstanding_ll}, all values of $d$ are $\bot$ for each process at point $t$.
  From \autoref{lem:idle_protection}(\ref{item:di_iff_protected}), no process strongly protects a tag at point $t$.
  The theorem follows from Lemmas~\ref{lem:unprotected} and~\ref{lem:quiescent_zero_tag}.
\end{proof}

\subsection{Correctness}

We conclude the analysis by proving, in \autoref{thm:linearizable}, that the implementation is linearizable.
We prove this last because the \confonly{full} proof relies on several lemmas used to show complexity and history independence, making linearizability the most involved proof overall.

\begin{lemma}\label{lem:protected_refcount}
  If tag $g \neq 0$ is protected for some object at point $t$, then $R[g] \neq 0$ at point $t$.
\end{lemma}
\fullonly{
  \begin{proof}
    Suppose that process $p$ protects tag $g \neq 0$ for some object $i$ at point $t^\star$.
    Observe from \autoref{def:protect}(\ref{item:nonzero_protected}) that when a nonzero tag $g$ is protected in \autoref{line:ll:second_read} or~\ref{line:cl:second_read}, the process $p$ protecting it has incremented $R[g]$ (in line \autoref{line:ll:inc_refcount} or~\ref{line:cl:inc_refcount} respectively) and not yet decremented $R[g]$.

    In the case of \autoref{line:ll:second_read}, by \autoref{def:protect}(\ref{item:nonzero_protected}), $L[i]$ matches $p$'s variable pair $(val, tag)$ at point $t^\star$ (so $L[i].\Tag = tag = g$); therefore, $p$ returns in \autoref{line:ll:matching_tag}.
    The next time that $p$ decrements $R[g]$ is thus either within a \CL{i} call, in which case $g$ ceases to be protected for $i$ at that call's invocation, or in \autoref{line:ll:dec_refcount} of a subsequent invocation of \LL{i}.
    However, since only \CL{i} can reset $d[i]$ to $\bot$ (in \autoref{line:cl:clear_di}), if $p$ invokes \LL{i} after point $t^\star$ without invoking \CL{i} in between, then it invokes \CL{i} in \autoref{line:ll:reset_link}.
    This makes $g$ no longer protected for $i$ before $p$ decrements $R[g]$ in \autoref{line:ll:dec_refcount}.

    In the case of \autoref{line:cl:second_read}, $p$ does not decrement $R[g]$ until \autoref{line:cl:final_return}, at which point $p$'s protection of $g$ ends.

    Therefore, $p$ does not decrement $R[g]$ after point $t^\star$ until after $g$ is no longer protected for $i$.
    The lemma then follows from \autoref{lem:incremented_nonzero}.
  \end{proof}
}

\begin{lemma}[Protection Invariant]\label{lem:invariant}
  If tag $g$ is protected for object $i$ at point $t$, then
  \begin{enumerate}[(a)]
    \item $L[i].\Tag$ does not change from a value different from $g$ to $g$ at point $t$, and\label{item:invariant:protection}
    \item if $g$ is \emph{strongly} protected for $i$ at point $t$, then $L[i].\Tag$ also does not change from a value different from $g$ to zero at point $t$.\label{item:invariant:strong_protection}
  \end{enumerate}
\end{lemma}
\begin{proof}
  We prove this lemma by contradiction.
  Suppose that the invariant is false, and let $t$ be the first point at which it fails.
  Thus, at point $t$, there exists a tag $g$ that is protected for $i$, and $L[i].\Tag$ changes from a tag $g' \neq g$ to $g$, or to zero if $g$ is strongly protected for $i$.

  Suppose first that~(\ref{item:invariant:strong_protection}) fails at point $t$, and therefore that $g$ is strongly protected for $i$ at point $t$, but that $L[i].\Tag$ changes from some tag $g' \neq g$ to zero at that point.
  Only \autoref{line:cl:cas} can change $L[i].\Tag$ to zero, since the only other lines which can change $L[i].\Tag$, namely lines~\ref{line:sc:first_cas} and~\ref{line:sc:second_cas}, choose a tag that is nonzero (from \autoref{line:sc:choose_tag}).
  Thus, there is a process $p$ which executes a successful \CAS{} in \autoref{line:cl:cas} at point $t$, changing $L[i].\Tag$ from $g'$ ($p$'s local value of $tag$) to zero.
  The if-conditions that $p$ checks in lines~\ref{line:cl:second_read} and~\ref{line:cl:confirm_quiescence} immediately before its \CAS{} at point $t$ both return \True, since $p$ executes \autoref{line:cl:cas}.
  Therefore, there exist points $t_L < t_C < t$, such that $L[i].\Tag = g'$ at point $t_L$, and $C[i] = 0$ at point $t_C$.

  From \autoref{def:protect}(\ref{item:nonzero_protected}), $g'$ becomes protected for $i$ at point $t_L$, and it remains protected until $p$ executes \autoref{line:cl:final_return}, which is after point $t$.
  From \autoref{lem:protected_activity}, $g$ is not strongly protected for $i$ at point $t_C$.
  Since by assumption $g$ is strongly protected for $i$ at point $t > t_C$, $g$ becomes strongly protected for $i$ at some point $t^\star \in (t_C, t)$.
  By \autoref{def:protect}, $L[i].\Tag = g$ at point $t^\star$.
  Since $L[i].\Tag = g' \neq g$ just before point $t > t^\star$, there is a point $t^\dagger \in (t^\star, t)$ at which $L[i].\Tag$ changes from some value different from $g'$ (possibly $g$) to $g'$.
  But since $g'$ is protected for $i$ throughout $(t_L,t]$ (recall that $t_L < t_C < t^\star < t^\dagger < t$), point $t^\dagger$ violates~(\ref{item:invariant:protection}), contradicting that $t$ is the first point at which the invariant fails.

  It remains to derive a contradiction for the case where~(\ref{item:invariant:protection}) fails at point $t$, and therefore that $L[i].\Tag$ changes from some tag $g' \neq g$ to $g$ at point $t$, while $g$ is protected for $i$.
  Note that if $g = 0$, then it cannot become protected in \autoref{line:cl:second_read}, as the \CL{} method would exit in \autoref{line:cl:zero_tag_return}.
  Thus, if $g = 0$ is protected then it is protected during an \LL{} call, and therefore it is also strongly protected, so the contradiction follows from that just derived for~(\ref{item:invariant:strong_protection}).
  We can therefore assume without loss of generality that $g \neq 0$.

  From \autoref{lem:protect_during_sc}, there is a process $p$ which executes a successful \CAS{} on $L[i]$ as part of an \SC{i, \cdot} call at point $t$, while strongly protecting some tag $g^\star$ from the call's invocation until at least point $t$.
  Let $t_R < t$ be the last point before $t$ when $p$ reads $R$ in \autoref{line:sc:check_tag}, and therefore reads $R[g] = 0$ and breaks out of the loop.
  Using \autoref{lem:protected_refcount}, it follows that $g$ is not protected for $i$ at point $t_R$, and therefore that $g \neq g^\star$.
  Thus, there exists a point $t^\star \in (t_R, t)$ at which $g$ becomes protected for $i$.
  By \autoref{def:protect}, $L[i].\Tag = g$ at point $t^\star$.
  Since $L[i].\Tag = g' \neq g$ immediately before $t$, it follows that there exists another point $t^\dagger \in (t^\star, t)$ at which $L[i].\Tag$ changes from some value different from $g'$ (possibly $g$) to $g'$.
  (We thus have that $t_R < t^\star < t^\dagger < t$.)

  Suppose that $g' \neq 0$.
  By \autoref{lem:protect_during_sc}, $g' = g^\star$.
  In that case, $L[i].\Tag$ changes from some value different from $g^\star$ to $g^\star$ at point $t^\dagger$.
  Since $g^\star$ is protected for $i$ throughout $[t_R, t] \ni t^\dagger$, this violates~(\ref{item:invariant:protection}) at point $t^\dagger < t$.

  Otherwise, $g' = 0$.
  Thus, at point $t^\dagger$, $L[i].\Tag$ changes from a nonzero tag $g^\dagger$ to zero.
  If $g^\dagger \neq g^\star$, then this violates~(\ref{item:invariant:strong_protection}) at point $t^\dagger < t$, since $g^\star$ is strongly protected for $i$ at point $t^\dagger \in (t_R, t)$.
  If $g^\dagger = g^\star$, then since $L[i].\Tag = g \neq g^\star$ at point $t^\star$, there is a point in $(t^\star, t^\dagger)$ at which $L[i].\Tag$ changes from some value different from $g^\star$ (possibly $g$) to $g^\star$.
  Since $g^\star$ is protected for $i$ throughout $[t_R, t] \supset (t^\star, t^\dagger)$, this violates~(\ref{item:invariant:protection}) before point $t^\dagger < t$.

  Thus, all possible cases contradict that $t$ is the first point which violates the invariant.
\end{proof}

\begin{lemma}\label{lem:sc_changes_tag}
  If an \SC{i, \cdot} operation has a successful \CAS{} at point $t$, then $L[i].\Tag$ changes from some tag $g$ to a new tag $g' \notin \set{g, 0}$ at point $t$.
\end{lemma}
\begin{proof}
  The fact that the new tag $g' \neq 0$ follows from the fact that \autoref{line:sc:choose_tag} chooses a nonzero value for the new tag.
  Thus, if the old tag $g = 0$, then the lemma follows immediately.
  Suppose instead that $g \neq 0$, and therefore that the calling process $p$ executes the successful \CAS{} in \autoref{line:sc:first_cas}, and $g = d[i].\Tag$.
  To show the lemma, it suffices to prove that $g \neq g'$.

  From \autoref{lem:protect_during_sc}, since $g \neq 0$, $g$ is protected when $p$ breaks out of its loop in \autoref{line:sc:check_tag}, upon reading $R[g'] = 0$.
  By \autoref{lem:protected_refcount}, $g'$ is not protected when $p$ executes \autoref{line:sc:check_tag}, so $g \neq g'$.
\end{proof}

\begin{definition}[Function $lin$]\label{def:lin}
  We define a function $lin(c)$ from method calls $c$ in an execution to points in time, and later (in \autoref{thm:linearizable}) prove that they are correct linearization points.

  If $c$ is an \LL{} call, then $lin(c)$ depends on which line the method responds from.
  If it responds in \autoref{line:ll:first_zero_tag} or~\ref{line:ll:final_return}, then $lin(c)$ occurs at \autoref{line:ll:first_read}.
  If it responds in \autoref{line:ll:matching_tag} or~\ref{line:ll:second_zero_tag}, then $lin(c)$ occurs at \autoref{line:ll:second_read}.

  If $c$ is an \SC{} call, then $lin(c)$ occurs at its invocation if it returns in \autoref{line:sc:early_return}, and otherwise $lin(c)$ is the point of the last \CAS{} operation.
  That is, if the \CAS{} in \autoref{line:sc:first_cas} succeeds, then $lin(c)$ occurs at \autoref{line:sc:first_cas}, and if the first \CAS{} fails then $lin(c)$ occurs at \autoref{line:sc:second_cas}.

  If $c$ is a \CL{} call, then $lin(c)$ occurs at its invocation.

  If $c$ is a \Read{} call, then $lin(c)$ occurs at \autoref{line:read} (its only step).

  Finally, if $c$ is a \VL{} call, then $lin(c)$ occurs at its invocation if it returns in \autoref{line:vl:early_return}, and otherwise $lin(c)$ occurs at \autoref{line:vl:read}.

  For an execution with pending invocations, $lin(c)$ is undefined for pending \LL{} calls, and for other calls for which the line corresponding to $lin(c)$ has not yet been executed.
\end{definition}
\fullonly{Note that since $lin(c)$ is defined retroactively for \LL{} calls, it cannot be a strong linearization point~\cite{GHW2011a}.}

\begin{lemma}\label{lem:successful}
  Suppose that $c$ is either an \SC{i, u} or \VL{i} call by process $p$.
  The call returns \True, and changes $L[i].\val$ to $u$ at point $lin(c)$, in the case of \SC{i, u}, if and only if
  \begin{enumerate}[(a)]
    \item some tag becomes strongly protected for $i$ by $p$ at some point $t < lin(c)$, and remains strongly protected by $p$ throughout $(t, lin(c)]$, and\label{item:successful:strong_protection}
    \item no \SC{i, \cdot} call has a successful \CAS{} in \autoref{line:sc:first_cas} or~\ref{line:sc:second_cas} in the interval $(t, lin(c))$.\label{item:successful:no_others}
  \end{enumerate}
\end{lemma}
\begin{proof}
  We first prove the forward direction, so we first assume that $c$ returns \True.
  Since $c$ returns \True, it responds in either \autoref{line:sc:return} or~\ref{line:vl:return_true}.
  Since in particular it does not return in \autoref{line:sc:early_return} or~\ref{line:vl:early_return}, it means that $d[i] \neq \bot$ when $p$ invokes $c$.
  By \autoref{lem:idle_protection}(\ref{item:di_iff_protected}), $d[i].\Tag$ is strongly protected at $c$'s invocation, starting from some previous point $t < lin(c)$, and by \autoref{def:protect} it remains protected until after $lin(c)$.
  Thus,~(\ref{item:successful:strong_protection}) is satisfied for some point $t < lin(c)$.

  Let $g^\star$ be the tag that $p$ strongly protects at point $t$, i.e., its value of $d[i].\Tag$ in call $c$.
  Suppose for the purpose of proving a contradiction that a successful \CAS{} from an \SC{i, \cdot} call takes place in $(t, lin(c))$, and let $t'$ be the point of the first such \CAS{}.
  Note that the only other \CAS{} which can take place on $L[i]$ happens in \autoref{line:cl:cas}, which also sets $L[i].\Tag$ to zero.
  Since, by \autoref{def:protect}, $L[i].\Tag = g^\star$ at point $t$, it follows that $L[i].\Tag$ is either $g^\star$ or zero immediately before point $t'$.
  By \autoref{lem:sc_changes_tag}, $L[i].\Tag$ changes to a different and nonzero value at point $t'$, and by \autoref{lem:invariant}, $L[i].\Tag$ does not change to $g^\star$ at point $t'$ since $g^\star$ is protected at that point.
  Thus, $L[i].\Tag$ changes to some nonzero tag distinct from $g^\star$ at point $t'$.
  Applying \autoref{lem:invariant} to each successful \CAS{} on $L[i]$ in $(t', lin(c))$ (if any), we obtain that $L[i].\Tag$ can only change to other nonzero tags distinct from $g^\star$ throughout the interval.
  But, since we assume that $c$ returns \True, and $lin(c)$ occurs at \autoref{line:sc:first_cas},~\ref{line:sc:second_cas}, or~\ref{line:vl:read}, it follows from the code that $L[i].\Tag \in \set{g^\star, 0}$ immediately before $lin(c)$, which is a contradiction.
  Thus,~(\ref{item:successful:no_others}) follows, finishing the proof for the forward direction.

  It remains to show the reverse direction, so we assume that~(\ref{item:successful:strong_protection}) and~(\ref{item:successful:no_others}) both hold for $c$.
  Like the forward case, let $g^\star$ be the tag that $p$ strongly protects throughout $(t, lin(c)]$, and we obtain from \autoref{def:protect} that $L[i].\Tag = g^\star$ at point $t$.
  From \autoref{lem:idle_protection}(\ref{item:di_iff_protected}), $p$'s value of $d[i].\Tag$ is $g^\star$ at $c$'s invocation, and in particular $p$ does not return in \autoref{line:sc:early_return} or~\ref{line:vl:early_return}.
  From~(\ref{item:successful:no_others}), no successful \CAS{} from an \SC{i, \cdot} call takes place throughout $(t, lin(c))$.
  The only other line that can change $L[i]$ is \autoref{line:cl:cas}, which changes $L[i].\Tag$ to zero while keeping $L[i].\val$ unchanged.
  It follows that $L[i].\Tag \in \set{g^\star, 0}$ immediately before point $lin(c)$, and that $L[i].\val$ is unchanged throughout $(t, lin(c))$.
  From \autoref{lem:idle_protection}(\ref{item:di_Li_match_when_protected}), $p$'s value of $d[i].\val$ equals $L[i].\val$ from $c$'s invocation until just before $lin(c)$, since $c$'s invocation is in the interval $(t, lin(c))$.

  If $c$ is a \VL{i} call, then the call returns \True because $g^\star$ is $p$'s value of $d[i].\Tag$.
  If $c$ is an \SC{i, \cdot} call and its \CAS{} in \autoref{line:sc:first_cas} succeeds, then the call returns \True.
  Otherwise, the \CAS{} in \autoref{line:sc:first_cas} fails, so $p$ performs a second \CAS{} in \autoref{line:sc:second_cas}, and $lin(c)$ is the point of that second \CAS{}.
  Since the only possible transition of $L[i].\Tag$ in $(t, lin(c))$ is from $g^\star$ to zero, it follows that $L[i].\Tag = 0$ when $p$ executes \autoref{line:sc:first_cas}, and therefore remains at zero immediately before $lin(c)$.
  Since $L[i].\val = d[i].\val$ just before $lin(c)$, the second \CAS{} succeeds.
\end{proof}

\begin{lemma}\label{lem:ll_strong_protection}
  Suppose that $c$ is an \LL{i} call by process $p$, such that $lin(c) = t$ and that the call responds at point $t'$.
  If some tag becomes strongly protected for object $i$ by $p$ during the call, then it becomes strongly protected at point $t$.
  Otherwise, some process $q \neq p$ has a \SC{i, \cdot} call $c'$ with a successful \CAS{} at point $lin(c')$, such that $t < lin(c') < t'$.
\end{lemma}
\fullonly{
  \begin{proof}
    Combining Definitions~\ref{def:protect} and~\ref{def:lin}, an \LL{} call $c$ that responds in \autoref{line:ll:first_zero_tag},~\ref{line:ll:matching_tag}, or~\ref{line:ll:second_zero_tag} causes a tag to become strongly protected at $lin(c)$, and an \LL{} call that responds in \autoref{line:ll:final_return} does not strongly protect any tag, proving the first claim.
    If $c$ responds in \autoref{line:ll:final_return}, then $p$'s local variable $tag \neq 0$ and $tag' \notin \set{tag, 0}$.
    In particular, at point $t^\star \in (t, t')$, $p$ executes \autoref{line:ll:second_read}, in which it reads $tag'$ from $L[i].\Tag$.
    Thus, there is some point $t^\dagger \in (t, t^\star)$ at which $L[i].\Tag$ first changes to a value not in $\set{tag, 0}$.
    The only lines which can change $L[i].\Tag$ to a nonzero value are a successful \CAS{} in lines~\ref{line:sc:first_cas} and~\ref{line:sc:second_cas}, and by \autoref{def:lin} such a successful \CAS{} operation is point $lin(c')$ of an \SC{i, \cdot} call $c'$.
    Hence, $t^\dagger = lin(c')$ for some \SC{i, \cdot} call $c'$ with a successful \CAS{} operation.
    Moreover, since $c'$ and $c$ overlap, it follows that $c'$'s calling process $q \neq p$.
  \end{proof}
}
\confonly{This lemma follows from the combination of Definitions~\ref{def:protect} and~\ref{def:lin}; note that the case of no tag becoming strongly protected corresponds to the case of returning in \autoref{line:ll:final_return}.}

\begin{theorem}[Linearizability]\label{thm:linearizable}
  The LL/SC implementation is linearizable according to the sequential specification, with $lin(c)$ as the linearization point for each call $c$.
\end{theorem}
\fullonly{
  \begin{proof}
    Consider an execution of LL/SC objects using this implementation.
    For each pending call $c$ in the execution such that $lin(c)$ is undefined according to \autoref{def:lin}, we remove $c$ from the linearization.
    For each \CL{} call invoked internally from an \LL{}, \SC{} or \VL{} call, we also remove it from the linearization.
    We construct the linearization from the remaining calls, sorting them by $lin(c)$.

    We prove that this linearization is correct, by induction over calls $c$ sorted by $lin(c)$.
    Specifically, we prove that each call $c$ returns a valid response and correctly affects the interpreted abstract state of the system.
    We inductively show that the interpreted value of LL/SC object $i$ is $L[i].\val$ before and after each linearization point.
    This is true initially since $L[i].\val$ is initialized to $u_i$, the application-dependent initial value for LL/SC object $i$.
    Recall from the sequential specification (\autoref{sec:seq_spec}) that the context bit $\ell_p$ of object $i$ equals $1$ (i.e., $p$ has an open link to $i$) immediately before a given method call $c$ of a sequential execution if and only if
    \begin{itemize}
      \item an \LL{i} operation $c'$ is called by process $p$ before call $c$ in the sequential execution,
      \item no successful \SC{i, \cdot} call (i.e., a call returning \True) by any process takes place in between $c'$ and $c$, and
      \item no direct \CL{i} call by $p$ takes place in between $c'$ and $c$.
    \end{itemize}
    It also follows from the sequential specification that the context bit only directly affects the return value of \SC{i, \cdot} and \VL{i} calls.
    As such, to prove correct behaviour of context bits, it suffices to show that \SC{i, \cdot} and \VL{i} calls return $\True$ if and only if the relevant context bit should be $1$ immediately before the respective call in the linearization.

    If $c$ is a \Read{i} or \LL{i} call, then its correctness follows from the fact that the call returns $L[i].\val$ at $lin(c)$.

    If an \SC{i, u} call $c$ returns \True, then $L[i].\val = u$ immediately after $lin(c)$, since by \autoref{def:lin} $c$ has a successful \CAS{} changing $L[i].\val$ to $u$ at point $lin(c)$.
    It also follows from the code that this is the only situation in which $L[i].\val$ can change.
    Therefore, if $c$ is an \SC{i, u} or \VL{i} call by process $p$, we only need to show that it returns \True if and only if the interpreted context bit $\ell_p = 1$, based on the three conditions listed earlier.

    Suppose that an \SC{i, u} or \VL{i} call $c$ by process $p$ does return \True, and we will show that $\ell_p = 1$ before $c$ according to the three conditions.
    By \autoref{lem:successful}(\ref{item:successful:strong_protection}), there is a point $t < lin(c)$ at which some tag $g$ becomes strongly protected for $i$ by $c$'s calling process $p$, and it remains strongly protected until at least $lin(c)$.
    By \autoref{def:protect}, only an \LL{i} call $c'$ by $p$ can cause $p$ to strongly protect $g$, and by \autoref{lem:ll_strong_protection}, $g$ becomes strongly protected at point $lin(c')$.
    By \autoref{lem:successful}(\ref{item:successful:no_others}), no \SC{i, \cdot} call has a successful \CAS{} in the interval $(lin(c'), lin(c))$.
    Since an \SC{i, \cdot} call $s$ with a successful \CAS{} returns \True and has $lin(s)$ at the point of that \CAS{}, it follows that there is no such call $s$ with $lin(c') < lin(s) < lin(c)$.
    Since a \CL{i} call by $p$ ends its strong protection of any tag for $i$, and $g$ remains strongly protected from $lin(c')$ until at least $lin(c)$, it follows that $p$ does not call \CL{i} in between $c'$ and $c$.
    Therefore, the three conditions are satisfied and the interpreted context bit $\ell_p = 1$ immediately before $lin(c)$, as required.

    Now suppose that the interpreted context bit $\ell_p = 1$ immediately before call $lin(c)$ for an \SC{i, u} or \VL{i} call $c$ by process $p$.
    Then, by the three conditions, $p$ has an \LL{i} call $c'$ with $lin(c') < lin(c)$, such that no successful \SC{i, \cdot} call $s$ has $lin(s) \in (lin(c'), lin(c))$, and $p$ does not directly call \CL{i} in between $c'$ and $c$.
    In particular, there is no successful \CAS{} on $L[i]$ directly from an \SC{i, \cdot} call in the interval $(lin(c'), lin(c))$.

    Suppose without loss of generality that $c'$ is the last \LL{i} call by $p$ before $c$ (if the three conditions hold for an earlier \LL{i} call, then they also hold for $c'$).
    From \autoref{lem:ll_strong_protection}, some tag $g$ becomes strongly protected by $p$ at point $lin(c')$ (since $p$ responds from $c'$ before it invokes $c$).
    In order to use \autoref{lem:successful}, we need to show that this strong protection lasts until at least $lin(c)$, and in particular that $p$ does not indirectly call \CL{i} in between $lin(c')$ and $lin(c)$.
    (Recall that indirect calls of \CL{} are removed from the linearization).
    The possible indirect calls of \CL{i} are in lines~\ref{line:ll:reset_link},~\ref{line:ll:clear},~\ref{line:sc:clear}, and~\ref{line:vl:clear}, so we will show that $p$ does not execute any of these lines in between $lin(c')$ and $lin(c)$ by deriving contradictions for each case.

    For \autoref{line:ll:reset_link}, since the line occurs before any possible line mapped to by $lin$ for \LL{}, it would have to be part of an \LL{i} call after $c'$, which contradicts that $c'$ is the last \LL{i} call by $p$ before $c$.
    For \autoref{line:ll:clear}, by the above arguments it should be part of call $c'$, but by \autoref{def:protect} and inspection of the code, no tag becomes strongly protected during $c'$, which contradicts that some tag does become strongly protected at point $lin(c')$.
    Line~\ref{line:sc:clear} would be part of an \SC{i, \cdot} call, so we will show that $p$ does not make any such call in between $c'$ and $c$.
    If such a call were to return \True, this contradicts the assumption there is no \SC{i, \cdot} call which returns \True in the linearization in between calls $c'$ and $c$.
    If such a call were to return \False, then by the inductive hypothesis there is another \SC{i, \cdot} call $s$ which returns \True with $lin(c') < lin(s) < lin(c)$, and this also contradicts our assumption.
    This last argument also applies if $p$ executes \autoref{line:vl:clear} in between $lin(c')$ and $lin(c)$, since that line is only executed in a \VL{i} call which returns \False.

    We therefore have that a tag becomes strongly protected by $p$ for $i$ at point $lin(c')$ and remains so until at least point $lin(c)$, and that no \CAS{} from an \SC{i, \cdot} call succeeds in between $lin(c')$ and $lin(c)$.
    Thus, from \autoref{lem:successful}, $c$ does return \True.

    Since \CL{} only affects $\ell_p$, which itself only becomes visible via \SC{} and \VL{} calls which we have already covered, the proof is complete.
  \end{proof}
}
\confonly{
  The linearizability proof from the full version~\cite{BHW2026a-full} is technical and too long to fit given the space constraints.
  The key ideas follow from Lemmas~\ref{lem:successful} and~\ref{lem:ll_strong_protection}, and the proof shows how ordering operations by $lin$ satisfies the sequential specification from \autoref{sec:seq_spec}.
}
}

\section{Implemented LL/SC in Existing History-Independent Algorithms}\label{sec:llsc_in_ABFOS}
We here briefly explain how to use our implemented LL/SC in the previously published history-independent universal construction~\cite{ABFOS2024a} and history-independent hash table~\cite{ABFOS2025a}.

In the state QHI universal construction's pseudocode (Algorithm 1 of~\cite{ABFOS2024a}), all blue lines are only required to achieve a wait-free universal construction from lock-free LL/SC\@.
Since our LL/SC implementation is randomized wait-free, the blue lines can be safely omitted, making the universal construction randomized wait-free.
Their algorithm also uses unconditional \Store{} operations, but they can be replaced by \LL{} followed by \SC{}, since this particular \SC{} is guaranteed to succeed.
The \RL{} calls in red lines become \CL{} calls with our implementation, except that the `if' condition in line 22 would need to be removed, i.e., that line becomes an unconditional \CL{announce[j]} call.

This last difference comes from the fact that their LL/SC implementation only needs to store context bits (i.e., open links) in addition to the LL/SC value, whereas ours also uses FAID counters.
When $a \neq \bot$ in line 22, the calling process $p_i$ knows that its open link to $announce[j]$ will be cleared by process $p_j$'s final \Store{} before the system next enters a quiescent configuration.
Since $p_j$'s \Store{} removes all traces of $p_i$'s open link, an \RL{} call is unnecessary under their lock-free LL/SC\@.
Our implementation has a slightly stronger requirement that processes must clear their own outstanding \LL{} operations, even if their open links are cleared, in order to ensure that elements of $C$, $R$, and $\Tag$ fields of $L$ reach zero.

For the history-independent hash table~\cite{ABFOS2025a}, \CL{i} can be added in the \FuncSty{insert()}, \FuncSty{lookup()}, \FuncSty{delete()}, and \FuncSty{help\_op()} methods immediately before each line which increments or decrements the variable $i$.
In each of these cases, the outstanding $\LL{i}$ is no longer needed once $i$ is changed.
Immediately before each response point (including responses from internal \FuncSty{help\_op()} calls), one can verify that the responding process has outstanding \LL{} operations for at most \LL{i-1}, \LL{i}, and \LL{i+1}, so \CL{} operations can be added for those three objects.
As discussed in~\autoref{sec:prelims:hi}, using an implemented LL/SC object loses the state QHI property for the hash table, making it only QHI\@.

\section{Conclusion}
In this paper we presented a new randomized implementation of multiple LL/SC objects from FAID and CAS that advances the state of the art in two ways:
First, it is the most space-efficient solution that achieves expected constant step complexity without relying on unbounded sequence numbers.
Second, it is the first wait-free algorithm that preserves quiescent history independence in concurrent data structures.

It is perhaps surprising that QHI-preserving LL/SC can be achieved not only through inefficient constructions, but also via an algorithm that achieves the best known asymptotic bounds in both step and space complexity.

We believe that the improvement on space complexity is an important complexity theoretical contribution, and that the algorithm is useful in real applications.
In particular, our algorithm renders the history-independent hash table of~\cite{ABFOS2025a} practically implementable on hardware (as it relies on LL/SC, which is not available in hardware) without an impact on its asymptotic efficiency.
While there are not yet many QHI algorithms that use LL/SC, we believe it is applicable to many future algorithms, because non-QHI algorithms typically do not need method calls of some operations to rely on \LL{}'s from other operations.

Our work complements, rather than replaces, the algorithms of Blelloch and Wei~\cite{BW2020a} and Jayanti and Petrovic~\cite{JP2003a}.
In contrast to their approaches, our construction requires augmenting LL/SC values with small bounded tags stored in CAS objects and, unlike the algorithm of Blelloch and Wei, does not support multi-word LL/SC\@.
Achieving these properties simultaneously with our space bounds remains an important open problem.

Other open problems include finding lower bounds on space complexity for deterministic or even randomized algorithms, as well as a general construction for inserting \CL{} operations into implementations using LL/SC while preserving their asymptotic space complexity.
Another direction would be to determine whether our space complexity and/or QHI-preserving property is achievable with a strongly linearizable implementation, with a randomized algorithm under the strong adaptive adversary, with a deterministic algorithm, or under different attack models.

\bibliographystyle{plainurl}
\bibliography{pwliterature}

\confonly{
  \appendix
  
}

\end{document}